\documentclass[sigconf, nonacm]{acmart}

\usepackage[T1]{fontenc}
\usepackage{amsmath}
\usepackage{mathtools} 
\usepackage{algorithm}
\usepackage{algorithmicx}
\usepackage[noend]{algpseudocode}
\usepackage{caption}
\usepackage{float}
\usepackage{multirow}
\usepackage{tabularx}
\usepackage{booktabs}
\usepackage{enumitem}
\usepackage{xcolor}
\usepackage[normalem]{ulem} 
\usepackage{tikz}
\usepackage[framemethod=TikZ]{mdframed}
\usepackage{xspace}
\usepackage{balance} 
\usepackage{bbm}
\usepackage{multicol}
\usepackage{filecontents}
\usepackage{tcolorbox}
\usepackage{tablefootnote}
\usepackage{booktabs}
\definecolor{ao}{rgb}{0.0, 0.5, 0.0}
\usepackage{listings,multicol}  
\usepackage{filecontents}
\usepackage{soul}
\usepackage{tcolorbox}
\newcommand{\code}[1]{{\texttt{#1}}}
\usepackage{microtype}  
\usepackage{hyphenat}   

\usepackage{subcaption}

\newtoggle{notes}
\toggletrue{notes} 

\newcommand{\cw}[1]{{\color{blue}\footnotesize[cw: #1]}}

\newcommand{\re}[1]{#1}

\newcommand{\boldparagraph}[1]{\vspace{7pt}\noindent\textbf{#1}}
\newcommand{\eat}[1]{}

\newcommand{\sys}{\textsc{SPECIAL}\xspace}

\newcommand{\planner}{\textsc{SPEplan}\xspace}

\newcommand{\pidx}{\textsc{SPEidx}\xspace}

\newcommand{\pop}{\textsc{SPEop}\xspace}

\newcommand{\card}{\textsc{SPEce}\xspace}

\newcommand{\ce}{\mathtt{CardEst}\xspace}
\newcommand{\synop}{\mathtt{synop}\xspace}

\newcommand{\vi}{\ensuremath{\mathsf{VIEW}}\xspace}
\newcommand{\negl}{\ensuremath{\mathsf{negl}}\xspace}

\newcommand{\lkg}{\ensuremath{\mathsf{Lkg}}\xspace}

\newcommand{\pp}{\ensuremath{\mathsf{pp}}\xspace}

\makeatletter
\DeclareFontFamily{OMX}{MnSymbolE}{}
\DeclareSymbolFont{MnLargeSymbols}{OMX}{MnSymbolE}{m}{n}
\SetSymbolFont{MnLargeSymbols}{bold}{OMX}{MnSymbolE}{b}{n}
\DeclareFontShape{OMX}{MnSymbolE}{m}{n}{
    <-6>  MnSymbolE5
   <6-7>  MnSymbolE6
   <7-8>  MnSymbolE7
   <8-9>  MnSymbolE8
   <9-10> MnSymbolE9
  <10-12> MnSymbolE10
  <12->   MnSymbolE12
}{}
\DeclareFontShape{OMX}{MnSymbolE}{b}{n}{
    <-6>  MnSymbolE-Bold5
   <6-7>  MnSymbolE-Bold6
   <7-8>  MnSymbolE-Bold7
   <8-9>  MnSymbolE-Bold8
   <9-10> MnSymbolE-Bold9
  <10-12> MnSymbolE-Bold10
  <12->   MnSymbolE-Bold12
}{}

\let\llangle\@undefined
\let\rrangle\@undefined
\DeclareMathDelimiter{\llangle}{\mathopen}%
                     {MnLargeSymbols}{'164}{MnLargeSymbols}{'164}
\DeclareMathDelimiter{\rrangle}{\mathclose}%
                     {MnLargeSymbols}{'171}{MnLargeSymbols}{'171}
\makeatother

\newcounter{theo}[section] 
\renewcommand{\thetheo}{\arabic{section}.\arabic{theo}}

\newcounter{lekg}[section] 
\newcommand{\thelkg}{\arabic{section}.\arabic{lekg}}

\newcounter{prot}[section] 
\renewcommand{\theprot}{\arabic{section}.\arabic{prot}}

\newcommand\vldbdoi{10.14778/3717755.3717764}
\newcommand\vldbpages{1035 - 1048}
\newcommand\vldbvolume{18}
\newcommand\vldbissue{4}
\newcommand\vldbyear{2024}
\newcommand\vldbauthors{\authors}
\newcommand\vldbtitle{\shorttitle} 
\newcommand\vldbavailabilityurl{URL_TO_YOUR_ARTIFACTS}
\newcommand\vldbpagestyle{empty} 

\begin{document}
\title{\sys: \re{S}yno\re{P}sis Assist\re{E}d Secure \re{C}ollaborat\re{I}ve \re{A}na\re{L}ytics}
\settopmatter{authorsperrow=4}

\author{Chenghong Wang}
\affiliation{%
  \institution{Indiana University}
}
\email{cw166@iu.edu}

\author{Lina Qiu}
\affiliation{%
  \institution{Boston University}
}
\email{qlina@bu.edu}

\author{Johes Bater}
\affiliation{%
  \institution{Tufts University}
}
\email{johes.bater@tufts.edu}

\author{Yukui Luo}
\affiliation{%
  \institution{Umass Dartmouth}
}
\email{yluo2@umassd.edu}

\begin{abstract}
Secure collaborative analytics (SCA) enables the processing of analytical SQL queries across data from multiple owners, even when direct data sharing is not possible. While traditional SCA provides strong privacy through data-oblivious methods, the significant overhead has limited its practical use. Recent SCA variants that allow controlled leakages under differential privacy (DP) strike balance between privacy and efficiency but still face challenges like unbounded privacy loss, costly execution plan, and lossy processing.

To address these challenges, we introduce \sys, the first SCA system that simultaneously ensures bounded privacy loss, advanced query planning, and lossless processing. \sys employs a novel {\em synopsis-assisted secure processing model}, where a one-time privacy cost is used to generate private synopses from owner data. These synopses enable \sys to estimate compaction sizes for secure operations (e.g., filter, join) and index encrypted data without additional privacy loss. These estimates and indexes can be prepared before runtime, enabling efficient query planning and accurate cost estimations. By leveraging one-sided noise mechanisms and private upper bound techniques, \sys guarantees lossless processing for complex queries (e.g., multi-join). Our comprehensive benchmarks demonstrate that \sys outperforms state-of-the-art SCAs, with up to $80\times$ faster query times, $900\times$ smaller memory usage for complex queries, and up to $89\times$ reduced privacy loss in continual processing.



\end{abstract}

\maketitle

{
\pagestyle{\vldbpagestyle}
\begingroup\small\noindent\raggedright\textbf{PVLDB Reference Format:}\\
\vldbauthors. \vldbtitle. PVLDB, \vldbvolume(\vldbissue): \vldbpages, \vldbyear.\\
\href{https://doi.org/\vldbdoi}{doi:\vldbdoi}
\endgroup
\begingroup
\renewcommand\thefootnote{}\footnote{\noindent
This work is licensed under the Creative Commons BY-NC-ND 4.0 International License. Visit \url{https://creativecommons.org/licenses/by-nc-nd/4.0/} to view a copy of this license. For any use beyond those covered by this license, obtain permission by emailing \href{mailto:info@vldb.org}{info@vldb.org}. Copyright is held by the owner/author(s). Publication rights licensed to the VLDB Endowment. \\
\raggedright Proceedings of the VLDB Endowment, Vol. \vldbvolume, No. \vldbissue\ %
ISSN 2150-8097. \\
\href{https://doi.org/\vldbdoi}{doi:\vldbdoi} \\
}\addtocounter{footnote}{-1}\endgroup


\ifdefempty{\vldbavailabilityurl}{}{
\begingroup\small\noindent\raggedright\textbf{PVLDB Artifact Availability:}\\
The source code, data, and/or other artifacts have been made available at \url{https://github.com/lovingmage/caplan}.
\endgroup
}
}

\fancypagestyle{plain}{%
  \fancyhf{} 
}
\pagestyle{plain} 

\section{Introduction} 
Organizations, such as hospitals, frequently hold sensitive data in separate silos to comply with privacy laws, despite the valuable insights that could be gained from sharing this information. Recent advancement of Secure Collaborative Analytics (SCA)~\cite{eskandarian2017oblidb, qin2022adore, roy2020crypt, wang2021dp, wang2022incshrink, bater2017smcql, bater2018shrinkwrap, liagouris2023secrecy, poddar2021senate, bater2020saqe, he2015sdb} provides an exciting solution to tackle this dilemma. These systems leverage advanced multi-party secure computation (MPC)~\cite{yao1986generate} primitives to empower multiple data owners, who previously could not directly share data, to collaboratively process analytical queries over their combined data while ensuring the privacy of each individual's data.




While MPC can effectively conceal data values~\cite{yao1986generate}, its security guarantees do not immediately extend to the protection of execution transcripts. Consequently, data-dependent processing patterns such as memory traces and read/write volumes can still reveal critical information, risking privacy breaches~\cite{cash2015leakage, kellaris2016generic, blackstone2019revisiting, oya2021hiding, grubbs2018pump, shang2021obfuscated, zhang2016all} even when the core data remains encrypted. To ensure strong privacy, modern SCA systems utilize {\em data-oblivious} primitives that exhaustively pad query processing complexities to a worst-case and data-independent upper bound~\cite{bater2017smcql, liagouris2023secrecy, poddar2021senate}. However, such stringent protections can largely reduce system efficiency and hinder the generalization of conventional optimization techniques to SCA, which are typically data-dependent~\cite{liagouris2023secrecy}. To address this, recent efforts~\cite{qin2022adore, wang2021dp, wang2022incshrink, bater2018shrinkwrap} have introduced Differentially Private SCA (DPSCA). This approach allows controlled information leakage under DP~\cite{dwork2014algorithmic} to mitigate constant worst-case overhead. For instance, systems under this model can dynamically compact an intermediate query size to a noisy estimate close to the actual size, avoiding exhaustive padding. As such, queries under DPSCA experience largely boosted efficiencies (e.g., up to $10^5\times$ faster~\cite{wang2022incshrink}) compared to their ``no leakage'' counterparts. Despite these substantial performance gains, existing DPSCAs still face critical limitations that impede their practical uses, as elaborated below:



\vspace{2.5pt}\noindent {\em $\bullet$ L-1. Unbounded privacy loss.}  Most DPSCA systems utilize a per-operator privacy expenditure model~\cite{bater2018shrinkwrap, wang2021dp, qin2022adore, chu2021differentially, chan2022foundations, wang2022incshrink}, meaning each query operator (e.g., join, filter) independently consumes a portion of the privacy budget. This approach can lead to either unbounded privacy loss or the forced cessation of query responses upon budget exhaustion. To mitigate this, some studies~\cite{zhang2023longshot, bogatov2021epsolute, qiudoquet} propose private, locality-sensitive grouping, incurring a one-time privacy cost to pre-group data based on specific attributes. Subsequent queries on those attributes can be directly applied to a smaller subset and need no additional privacy budget.  However, this method only supports simple queries (e.g., point and range); complex queries like joins still suffer from unbounded privacy loss.

\vspace{2.5pt}\noindent{\em $\bullet$ L-2. \re{Unoptimized execution plan.}} Conventional query planners can pre-estimate sizes for equivalent plans of a given query and select the most efficient plan with minimized intermediate sizes before execution~\cite{silberschatz2011database, blasgen1981system}. In contrast, SCA systems lack this capability, and even DPSCA designs~\cite{qin2022adore, wang2022incshrink, zhang2023longshot, wagh2021dp, bater2018shrinkwrap, bater2020saqe} can only reactively determine plan sizes during runtime. This inherent limitation often forces existing systems to settle for less efficient query plans, such as costly join orders, which lead to significantly inflated intermediate sizes (\S~\ref{sec:ete}) and substantially hinder performance.
    
\vspace{2.5pt}\noindent {\em $\bullet$ L-3. Lossy processing.} Noise from randomized mechanisms in DPSCA also introduces a unique accuracy issue \re{(e.g., conventional DP mechanisms may generate negative noise, applying which to obfuscate the sizes of intermediate query results can cause losing qualified real tuples)}, and unfortunately, no existing DPSCA can mitigate such loss for complex queries~\cite{wang2021dp,wagh2021dp,wang2022incshrink,zhang2023longshot,groce2019cheaper}. Furthermore, stronger DP settings can further increase noise variance, which amplifies errors, significantly impacting the utility of SCA systems.

\subsection{Overview of \sys}\label{sec:overview}
In this work, \re{we introduce \sys, an innovative SCA system that resolves the aforementioned limitations all at once through a new paradigm called {\em synopsis-assisted secure processing}.} At its core, \sys incurs a one-time privacy cost to gather DP synopses (statistics of base tables) from owners' data. These synopses are then used to accelerate complex query processing, and enhance SCA query planning. Notably, \sys is the first system to provide all of the following benefits: (1) {\em Bounded privacy}—\re{the privacy loss in \sys is strictly limited to the one-time synopses release stage, with absotely no additional privacy cost during complex query processing and planning;} (2) {\em Advanced query planning}—it builds an advanced SCA planner that can exploit plan sizes before runtime; and (3) {\em Lossless processing}—it ensures exact results with no data omissions.
 An overview
of \sys is shown in Figure~\ref{fig:flow}.

\begin{figure}[h]
\centering
\includegraphics[width=0.65\linewidth,interpolate=false]{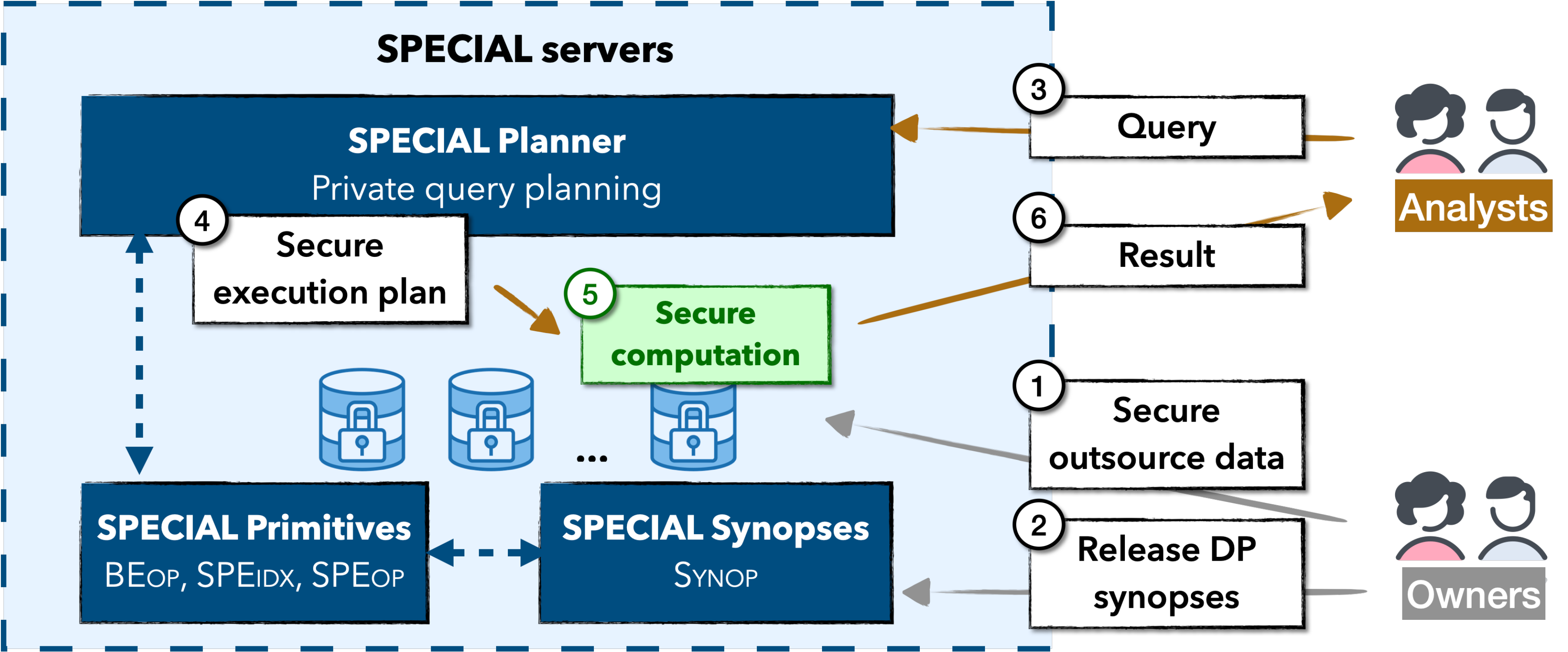}
\caption{Overview of \sys workflow.}
\label{fig:flow}
\end{figure}

\sys operates under a standard server-aided MPC model~\cite{kamara2011outsourcing} with three key participants: data owners, at least two \sys servers, and a vetted analyst. The process begins with data owners securely outsourcing their data, typically through secret sharing~\cite{kamara2011outsourcing, wang2021dp, zhang2023longshot, wang2022incshrink}, and privately releasing corresponding DP synopses (\S~\ref{sec:prim}) to the servers. \sys introduces a set of novel primitives (\S~\ref{sec:primitive}) that can leverage these synopses to accelerate secure query operations. 
Once the data and synopses are in place, analysts can submit \re{Select-Project-Join-Aggregation (SPJA) queries~\cite{silberschatz2011database}} for analytics. To process queries, a private planner (\S~\ref{sec:planner}), running on \sys servers, strategically orchestrates \sys primitives to process the query and optimize performance. Finally, the results are securely returned to the analyst.

\subsection{Unique challenges and key contributions}\label{sec:contributions} Leveraging DP synopses in SCA holds significant promise for achieving our desired objectives. However, this also introduces unique challenges. Below, we highlight the key challenges and summarize our non-trivial contributions to address them:

\label{challenge-c1}\vspace{2.5pt}\noindent {\em $\bullet$ C-1. How to select proper synopses?} Even for a single relation, one can find numerous attribute combinations for generating synopses.
Improper selection can lead to large errors (e.g. using too many synopses or high-dimensional attributes~\cite{zhang2021privsyn}), or reduced functionalities (e.g., using only simple attributes~\cite{zhang2023longshot}). Hence, a key challenge is selecting a limited set of DP synopses to optimize the privacy budget for complex query processing. Our approach is informed by two observations: (i) secure joins are resource-intensive and need prioritized acceleration, and (ii) synopses for common filtering predicates are vital as they allow pre-built indexes on base relations for fast access. Consequently, we propose a focused strategy (\S~\ref{sec:prim}) that targets low-dimensional (1D and 2D) attributes frequently involved in joins and filters within a representative workload.

\label{challenge-c2}\vspace{2.5pt}\noindent {\em $\bullet$ C-2. How to enforce lossless processing?} Private synopses do not immediately implies lossless guarantees. Thus, a second challenge is designing practical approaches to achieve lossless results without violating privacy goals. To address this, we employ one-sided DP noise (either strictly positive or negative, \S~\ref{sec:prim}) in generating synopses, and design novel primitives (\S~\ref{sec:primitive}) based on them to pessimistically estimate filter cardinalities and intervals of index structures. 
To ensure lossless processing of complex joins, we extend upon cutting-edge join upper bound techniques~\cite{hertzschuch2021simplicity} to privately estimate lossless join sizes using DP synopses. To our knowledge, this is the first study to support private join upper bound estimation. 

\label{challenge-c3}\vspace{2.5pt}\noindent {\em $\bullet$ C-3. How DP synopses can empower efficient query processing?} The use of DP synopses in SCA is largely underexplored, leaving a knowledge gap regarding their potential to enhance query efficiency. To navigate this potential, we explore various use of synopses in accelerating secure processing including private indexes \pidx (\S~\ref{subsec:dpindex}), and compacted oblivious operations \pop (\S~\ref{sec:eops}). \re{ We also design a novel private query planner (\S~\ref{sec:planner}) that efficiently orchestrate the execution of \sys primitives (e.g., \pidx, \pop) to process SPJA queries. The planner uses available synopses to privately estimate intermediate result sizes and operation costs for a set of equivalent execution plans of a given query. It then executes the one with the lowest estimated cost.}


\vspace{2.5pt}\noindent {\em $\bullet$ C-4. How to systematically evaluate \sys?} A major is the absence of open benchmarks. We address this by initiating an open-source evaluation set, accessible to the public. Specifically, we use public financial data~\cite{RelationalDataFinancial} and design eight test queries, ranging from simple linear queries to complex 5-way joins. We also re-produce an open version of the HealthLNK benchmark. We evaluate our prototype, \sys, against the state-of-the-art (SOTA) DPSCA system, Shrinkwrap~\cite{bater2018shrinkwrap}, and the conventional SCA system, SMCQL~\cite{bater2017smcql}.  Results indicate that \sys outperforms Shrinkwrap, reducing query latency by up to $80.3\times$, and SMCQL, with at least a $114\times$ reduction in query latency. Additionally, \sys improves memory efficiency in complex join processing by more than $900\times$ compared to both systems. Moreover, scaling experiments show that \sys can effectively scale up to 8.8M rows dataset and up to 9-way complex joins. All benchmarks, including our prototype implementation, are open-sourced and available at~\cite{special2024}.

\section{Background}\label{sec:def}
\noindent{\bf General notations.} We consider the logical database $\mathcal{D}$ to contain multiple private (base) relations \(\{D_1, D_2,...\}\), where each relation $D_i$ is owned by a specific party $P_i$. A base relation \(D\) (we omit subscript for simplicity) has a set of attributes \(attr(D)\). The domain of an attribute \(A \in attr(D)\) is denoted by \(dom(A)\), and the combined domain of a collection of attributes \(\mathbf{A} = \{A_1, A_2,...\} \subseteq attr(D)\) is denoted by \(dom(\mathbf{A}) = \prod_{A\in\mathbf{A}}dom(A)\). For a tuple $t\in D$, and $\mathbf{A}\subseteq attr(D)$, we use $t.\mathbf{A}$ to denote the attribute value of $\mathbf{A}$ in $t$. \re{A logical query, represented by $q(\mathcal{D})$, applies transformations and computations on $\mathcal{D}$. In this work, we focus on SPJA~\cite{silberschatz2011database} queries. }\\ 
\smash{\underline{Frequency (count).}} Given $D$, $\mathbf{A}\subseteq attr(D)$, and a set of values $\mathbf{v} \in dom(\mathbf{A})$, the frequency (count) of $\mathbf{v}$ in $D$ is the total number of tuples $t\in D$ with $t.\mathbf{A} = \mathbf{v}$. In addition, the max frequency moments (MF) of $\mathbf{A}$ is defined as ${\mathtt{mf}}(\mathbf{A}, D) = \max_{\mathbf{v}\in dom(\mathbf{A})}| \{ t\in D \mid t.\mathbf{A} = \mathbf{v}\}|$.\\
\smash{\underline{Histograms.}} Given $D$, and $\mathbf{A} \subseteq attr(D)$, the (equal-width) histogram $\mathbf{h}(\mathbf{A}, D) = (c_1, c_2,...,c_m)$ is a list of counts for the attribute values in $\mathbf{A}$. Specifically, $\mathbf{h}$ partitions $dom(\mathbf{A})$ into $m$ ``equal-sized'' domain intervals $(B_1,...,B_m)$, and a count $c_i \in \mathbf{h}$ is the number of tuples $t\in D$ with $t.\mathbf{A}$ in the interval of $B_i$. 

\re{\vspace{3pt}\noindent{\bf Query planning.} Modern databases parse queries into physical plans~\cite{silberschatz2011database} that can be executed by the underlying query engine. These plans specify the operations like scans, joins, and sorts, and the order in which they're performed. The same query can have multiple equivalent execution plans, but their performance can vary greatly depending on resource usage and data access patterns. Query planning~\cite{blasgen1981system}, done before runtime, involves selecting cost-efficient plans from these options. A key part of this process is accurately estimating intermediate result sizes, known as cardinality estimation (CE)~\cite{harmouch2017cardinality}, which relies heavily on table statistics. Two crucial statistics in modern cardinality estimation methods are: (i) {\em histograms,} which are vital for estimating selectivities in filters, and (ii) {\em max frequency,} which is crucial for estimating join sizes.}

\vspace{3pt}\noindent{\bf Multi-party secure computation (MPC).} MPC~\cite{yao1986generate, goldreich2009foundations, micali1987play, ben2019completeness} is a cryptographic technique that allows multiple parties $P_1, P_2,...$ to jointly compute a function $f(x_1, x_2,...)$ over their own private input $x_i$. MPC ensures no unauthorized information is revealed to any party, except the desired output of $f$, emulating a computation as if performed by a trusted third party. Traditional MPC required all parties to actively participate in intensive computations. However, recent server-aided MPC~\cite{kamara2011outsourcing, scholl2017s, mohassel2017secureml} schemes allow offloading computations to powerful servers, without sacrificing security. In this model, parties secretly share their inputs with servers, which jointly evaluate an MPC protocol to reconstruct the secrets and compute the function. 

\vspace{3pt}\noindent{\bf Differential privacy~\cite{dwork2014algorithmic}.} DP ensures that modifying a single input tuple to a mechanism produces only a negligible change in its output. To elaborate, consider $D$ and $D'$ as two relations differing by just one tuple, then DP defines the following.

\begin{definition}[$(\epsilon, \delta)$-DP]\label{bg:dp} \emph{Given $\epsilon>0$, and $\delta \in (0,1)$. A randomized mechanism $\mathcal{M}$ is said to be $(\epsilon, \delta)$-DP if for all ${D}\sim{D'}$ pairs, and any possible output $o\subset Range(\mathcal{M})$, the following holds:
$$\textup{Pr}\left[\mathcal{M}(D) \in o\right] \leq e^{\epsilon} \textup{Pr}\left[\mathcal{M}(D') \in o\right] + \delta$$}
\end{definition}




\vspace{2pt}\noindent{\bf Secret sharing and secure array.} \sys uses the $2$-out-of-$2$ boolean secret share~\cite{araki2016high} over ring $\mathbb{Z}_{2^{32}}$  for securely outsourcing owners' data and storing query execution results. Specifically, each data, $x$, is divided into two shares: $x1, x2$ that are uniformly distributed over the ring $\mathbb{Z}_{2^{32}}$ such that $x = x1 \oplus x2$. Each server $S_i$ receives one secret shares, $s_i$, where $i\in \{0,1\}$. By retrieving shares from any two servers, an authorized party can successfully reconstruct the value of $x$. However, a single server alone learns nothing about $x$. For clarity and to abstract out the lower-level details, we leverage a logically unified data structure, namely the secure array~\cite{wang2022incshrink, bater2018shrinkwrap}, denoted as $\langle \mathbf{x}\rangle = \left(\langle x_1\rangle, \langle x_2\rangle,...\right)$, which is a collection of secret-shared relational tuples.


\vspace{3pt}\noindent{\bf Oblivious (relational) operators.} Oblivious operators are data-independent MPC protocols that implement the same functionalities as their plaintext database counterparts (e.g., filter and join). Data-independent execution requires that the control flow and memory access patterns of a function are indistinguishable given different inputs of the same size, and typically requires costly computation. For example, a linear scan is required to fulfill oblivious filtering~\cite{zheng2017opaque}, and join requires nested-loop over the two inputs~\cite{eskandarian2017oblidb}. \re{The output sizes of such operators are usually padded with dummy tuples to the worst case: $N$ rows for filters and $N^2$ for joins, given size $N$ inputs. The dummy tuples will not affect the query result but can significantly impact the performance~\cite{zheng2017opaque, eskandarian2017oblidb}. To enhance efficiency while maintaining strong privacy, DPSCAs introduce a new type of oblivious operators~\cite{bater2018shrinkwrap, qin2022adore, qiudoquet, wang2022incshrink}. These operators typically involve two steps: {\em Compute} and {\em Compact}. The {\em Compute} step is fully oblivious, while the {\em Compact} resizes the output, often by obliviously sorting valid tuples to the front and trimming the output to a noisy DP size (the true size plus DP noise). While this approach can significantly reduce the computation cost and query sizes, it may lead to lossy query processing if the DP size is smaller than the true size (e.g., negative DP noise), as {\em valid tuples could be excluded during the compaction}~\cite{wang2021dp,wang2022incshrink, zhang2023longshot}. We emphasize that when DP sizes exceed true sizes, there is no accuracy loss as it only includes extra dummy data that do not impact accuracy~\cite{zheng2017opaque, eskandarian2017oblidb}.

\vspace{3pt}\noindent{\bf Private indexes.} In conventional databases, indexes are powerful data structures that map attribute values to positions in a sorted array, allowing a predicate selection to quickly access the desired data via index lookup without the need for full table scan. However, traditional indexes are unsuitable in SCAs due to their data-dependent nature, which can easily lead to privacy breaches. To address this, recent research has proposed DP indexes~\cite{roy2020crypt}, where the mapping of attribute values to their positions is intentionally distorted with DP noise. To process queries, the system first pre-fetch a small set of data using DP indexes, followed by oblivious selection. This effectively avoids full table oblivious scan and sorting-based result compaction. However, the uncertainty inherent in DP indexes can lead to the loss of valid tuples. For example, in the DP index of~\cite{roy2020crypt}, a true index range of positions $[10, 20]$ might be distorted into positions $[12, 17]$, causing data at position $10, 11, 18, 19$ and $20$ to be missed. Nevertheless, {\em if DP indexes overestimate the range, subsequent oblivious selection can losslessly identify all valid tuples.}
}



\section{System and Privacy Model}\label{sec:modeli}

\re{In general, \sys follows a standard server-aided MPC~\cite{kamara2011outsourcing} model, involving (i) a set of mutually distrustful data owners ${P_1,...,P_n}$, (ii) two non-colluding servers $S_0$ and $S_1$, and (iii) a trusted analyst. We assume an {\em admissible adversary}~\cite{mohassel2017secureml} $\mathcal{A}$, capable of corrupting $n-1$ out of $n$ clients and at most one of the two servers. An instance of such adversary can be a malicious server that creates Sybil owners to form a malicious collation, attempting to steal sensitive information from an honest owner. Additionally, $\mathcal{A}$ is considered honest-but-curious, meaning it follows the protocol without deviation but may try to infer information from observed protocol transcripts, such as randomness, memory access patterns, and communication messages. The combination of these information is referred to as the view of $\mathcal{A}$. We also assume $\mathcal{A}$ is computationally bounded as a probabilistic polynomial time (p.p.t.) adversary, which is a standard requirement in MPC protocols to ensure that adversaries cannot break cryptographic primitives. This threat model is consistent with prior SCA designs~\cite{bater2017smcql, bater2018shrinkwrap, wang2022incshrink, mohassel2017secureml, wang2021dp}. Given this setup, we design \sys to satisfy the following:
}
\begin{definition}[MPC protocol with DP leakage]\label{def:privacy-def}{\em Given a set of parties (owners) $P_i$ with private data $D_i$ and a secure query protocol $\Pi$ that applies over $\mathcal{D}=\{D_1,D_2,\cdots\}$. We define a randomized mechanism $\lkg(\mathcal{D})=\{\lkg(D_1),\lkg(D_2),\cdots\}$ as the leakage profile, consisting of the control flow and access patterns of running $\Pi$ over $\mathcal{D}$. The protocol $\Pi$ is said to be secure with DP leakage if, for the subset of uncorrupted parties with data $\mathbf{D}\subseteq\mathcal{D}$, leakage profile $\lkg(\mathbf{D})\subseteq\lkg(\mathcal{D})$, and any p.p.t. adversary $\mathcal{A}$:
\begin{itemize}[leftmargin=10pt]
    \item $\lkg(\mathbf{D})$ satisfies $(\epsilon, \delta)$-DP (definition~\ref{bg:dp}).
     \item There exists a p.p.t. simulator $\mathcal{S}$ with only access to public parameters $\pp$ and $\lkg(\mathbf{D})$ that satisfies:
     \begin{equation}
    \begin{split}
        & \textup{Pr}\left[\mathcal{A}\left(\vi^{\Pi}(\mathcal{D}, \pp)=1\right)\right] \\
        & \leq  \textup{Pr}\left[\mathcal{A}\left(\vi^{\mathcal{S}}(\lkg(\mathbf{D}), \pp)\right)=1\right] + \negl(\kappa)
    \end{split}
\end{equation}
\end{itemize}
where \re{$\vi^{\Pi}$ is $\mathcal{A}$'s view in $\Pi$'s execution and $\vi^{\mathcal{S}}$ is a simulated view produced by $\mathcal{S}$ using $\mathsf{Lkg}$; $\mathsf{pp}$ denotes all public parameters, and $\negl(\kappa)$ is a negligible function related to a security parameter $\kappa$.}
}
\end{definition}

\re{ Simply put, Definition~\ref{def:privacy-def} requires that the knowledge any p.p.t. adversary adversary can gain about each individual tuple of an honest owner, by observing the protocol execution, is bounded to what can be inferred from the outputs of the $(\epsilon, \delta)$-DP mechanism $\lkg$. We stress that this notion focuses on DP at the event (tuple) level without loss of generality. Due to the group-privacy properties of DP~\cite{dwork2010differential, kifer2011no, xiao2015protecting, vadhan2017complexity}, event-level DP can be extended to user-level DP. For instance, in a logical database $\mathcal{D}$ where any single user owns at most $l$ tuples, if a protocol satisfies $(\epsilon, \delta)$ event-level DP, it also satisfies $(l\epsilon, le^{(l-1)\epsilon}\delta)$ user-level DP. Moreover, we say that \sys can be relaxed to employ a weaker corruption model, such as requiring a supermajority of owners and servers to remain uncorrupted, to enhance efficiency~\cite{liagouris2023secrecy,tan2021cryptgpu}. This adjustment does not change the privacy guarantee outlined in Definition~\ref{def:privacy-def}, but it does affect the security assumptions. Under the relaxed corruption model, Definition~\ref{def:privacy-def} is only satisfied when at least two-thirds of the parties remain uncorrupted.}  Due to space concerns, we defer the complete privacy proof of \sys to our full version~\cite{special2024online}.

\section{$\sys$ Synopses}\label{sec:prim}
We now discuss the details of private synopses used in \sys, while in later sections we will show how they accelerate query processing (\S~\ref{sec:primitive}) and aid in query planning (\S~\ref{sec:planner}).



\vspace{3pt}\noindent{\bf Challenges.} We reiterate the main challenges in generating private synopses for a relation include: \hyperref[challenge-c1]{(C-1)} selecting a set of attribute combinations that enable functional and efficient query processing; \hyperref[challenge-c2]{(C-2)} ensuring that the subsequent query processing based on the private synopses is lossless.



\vspace{3pt}\noindent{\bf Key ideas.} Given a SPJA query, join operations typically need prioritized acceleration, as they are more resource-intensive than other operations. A $k$-way join can have $O(n^k)$ complexity without optimizations~\cite{bater2017smcql, bater2018shrinkwrap, wang2022incshrink, zheng2017opaque, eskandarian2017oblidb}. Additionally, synopses for frequently queried filter attributes play an important role in efficient query processing, as they enable fast indexing (\S~\ref{subsec:dpindex}) and effective filtering of unnecessary data before heavy joins. As such, our first key idea is to focus on histograms-based synopses that cover frequently queried join and filter attributes. To minimize noise, we focus on low-dimensional synopses: only 1D or 2D histograms.



To address the second challenge, our approach combines two strategies. First, to support accurate indexing and filtering, we use one-sided DP noise to generate DP histograms that consistently overestimate or underestimate attribute distributions. We will show later that such special histograms allow lossless filtering and indexing that reliably overestimate true filter sizes and indexing ranges (\S~\ref{subsec:dpindex}). Second, for lossless join output compaction, we incorporate noisy max frequency moments (MF) into the synopses. MF allows us to build on advanced join upper bound techniques~\cite{hertzschuch2021simplicity} to privately estimate join sizes without data loss (\S~\ref{subsec:dpindex}).

\subsection{Synopses generation}\label{sec:snop_gen} 
We now elaborate on the details of synopsis generation, which mainly contains two phases: (i) Attributes selection, where the \sys servers select appropriate attributes for the generation, which are then distributed to owners; (ii) Local synopses release, where the owners create corresponding synopses using a DP mechanism and upload them to \sys servers.


\vspace{3pt}\noindent{\bf Attributes selection (servers).} The first step is to identify a set of attributes for deriving synopses. In general, we consider the existence of a representative workload, $Q_{\mathsf{R}}$~\cite{kotsogiannis2019privatesql}, which can be sourced from a warm-up run or annotated by the administrator. Note that the representative workload does not involve any private data and thus is leakage-free. The servers first identify representative attribute pairs, \(\mathsf{pair} = \{\mathsf{pair}_k\}_{k\geq 1}\), for each private relation \(D \in \mathcal{D}\) via \(Q_{\mathsf{R}}\). The designated pairs include: (i) 2-way attribute pairs, which correspond to frequently queried \emph{filter-join key} combinations; (ii) frequently queried individual attributes not covered by these pairs. By default, each $\mathsf{pair}_k = (A_{\mathsf{ft}}, A_{\mathsf{j}})$ contains two valid attributes (case i), but either $A_{\mathsf{ft}}$ or $A_{\mathsf{j}}$ may be empty (case ii).

\vspace{3pt}\noindent{\bf Synopses release (owners).} Next, servers pushes the identified pairs to owners, and subsequently, the owners independently dispatche private synopses and return them to servers.  We now focus on the DP synopses generation mechanism run by each owner, Algorithm~\ref{algo:synopgen} illustrates the workflow.


\begin{algorithm}[]
\caption{DP synopsis gen $\mathcal{M}_{\synop}$ (in the view of $P$)}
\begin{algorithmic}[1]
\Statex
\textbf{Input}: $\mathsf{pair} = \{\mathsf{pair}_k\}_{k\geq 1}$ from servers; private data $D$.
\State $P$ self-determines privacy parameters $\epsilon, \delta$, and init $\synop \gets \emptyset$
\For{each $\mathsf{pair}_k$}
\State $\mathbf{h}(\mathsf{pair}_k, D) \gets \mathsf{HistGen} (\mathsf{pair}_{k}, D)$
\Statex \textcolor{gray}{\underline{DP histograms:}}
\State $\mathbf{h}^{+}(\mathsf{pair}_k, D) \gets \mathbf{h} + \mathsf{Lap}^{+}(\epsilon, \delta, \mathbf{h}.\mathsf{shape})$
\State $\mathbf{h}^{-}(\mathsf{pair}_k, D) \gets \mathbf{h} + \mathsf{Lap}^{-}(\epsilon, \delta, \mathbf{h}.\mathsf{shape})$
\Statex\Comment{{\em adding independently sampled noise to every bin of $\mathbf{h}^{+}$, $\mathbf{h}^{-}$} }
\Statex \textcolor{gray}{\underline{DP max frequencies:}}
\If{$A_{\mathsf{j}} \in \mathsf{pair}_k = \emptyset$ {\bf or} $A_{\mathsf{j}}$ is unique valued} ~ $\mathtt{MF}_{k} = \emptyset$
\ElsIf{$A_{\mathsf{ft}} \in \mathsf{pair}_k \neq \emptyset$}
\Statex \Comment{{\em assuming $\mathbf{h}$ partitions $dom(A_{\mathsf{ft}})$ into $\{B_1,...,B_m\}$}}
\State $D^{\ell} \gets \sigma_{A_{\mathsf{ft}} \in B_{\ell}}(D)$ {\bf for} $\ell = 1,2,...,m$
\State compute noisy MF table, $\mathtt{MF}_{k}=\{\widehat{\mathtt{mf}}(A_{\mathsf{j}}, D^{\ell})\}_{1\leq\ell\leq m}$
\Else{~\smash{$\mathtt{MF}_{k}=\widehat{\mathtt{mf}}(A_{\mathsf{j}}, D)$}}
\EndIf
\State \smash{$\synop \gets \synop \cup (\mathsf{pair}_k, \{{\bf h}^{+}, {\bf h}^{-}\}, \mathtt{MF}_{k})$}
\EndFor 
\State {\bf release} $\synop$, $\epsilon$, $\delta$ to servers
\end{algorithmic}
\label{algo:synopgen}
\end{algorithm} 
{
In general, we expect owners to set a desired privacy budget for their data (using parameters \( \epsilon \) and \( \delta \)). Algorithm~\ref{algo:synopgen} produces synopses formalized as :
\begin{definition}[\sys synopses] \label{def-synop}{\it Given \(Q_{\mathsf{R}}\), we consider for each relation \(D\), its corresponding synopsis $\synop$ is the collection of \(\{(\mathsf{pair}_k, \mathbf{H}(\mathsf{pair}_k, D), \mathtt{MF}_k)\}_{k\geq 1}\), such that
\begin{itemize}[leftmargin=10pt]
    \item $\mathsf{pair}_k = (A_{\mathsf{ft}}, A_{\mathsf{j}}) \in Q_{\mathsf{R}}$ is a frequently queried attribute pair.
    \item \(\mathbf{H}(\mathsf{pair}_k, D) = \{\mathbf{h}^{+}, \mathbf{h}^{-}\}\) is the DP bounding histogram for $\mathsf{pair}_k$, where $\mathbf{h}^{+}$ (resp. $\mathbf{h}^{-}$) is a DP histogram that overestimate (resp. underestimate) the true histogram of  $\mathsf{pair}_k$. 
    \item \(\mathtt{MF}_k\) represents a collection of privately overestimated join key MFs categorized by \(\mathsf{pair}_k.A_{\mathsf{ft}}\).
\end{itemize} }
\end{definition}
}

{
We now present a detailed explanation of generating these synopsis structures, starting with the private bounding histograms (Alg~\ref{algo:synopgen}, lines 2:5). Specifically, for each $\mathsf{pair}_k$, each owner first constructs a histogram $\mathbf{h}(\mathsf{pair}_k, D)$. By default, we assume there exists global parameters (e.g., bin sizes) for each attribute, that ensure consistent data partitioning among all owners. Next, the owner derives two noisy histograms, $\mathbf{h}^{+}(\mathsf{pair}_k, D)$ and $\mathbf{h}^{-}(\mathsf{pair}_k, D)$, by adding independently sampled one-sided Laplace noise (Definition~\ref{def:1side}) to every bin of $\mathbf{h}(\mathsf{pair}_k, D)$. This guarantees that \smash{\(\mathbf{h}^{+}\)} always overestimates the true histogram, while \smash{\(\mathbf{h}^{-}\)} consistently underestimates it.
 }

\begin{definition}[One-sided Laplace variable]\label{def:1side}\emph{$\mathsf{Lap}^{+}(\epsilon, \delta) = \max(0,z)$ (resp. $\mathsf{Lap}^{-}(\epsilon, \delta) = \min(0,z)$) is a one-sided Laplace random variable in the range of $[0, \infty)$ (resp. $(-\infty, 0]$)  if $z$ is drawn from a distribution with the following density function
\vspace{-.5mm}
\begin{equation}
    \begin{split}
        \textup{Pr}\left[z = x\right] = \frac{e^{\epsilon}-1}{e^{\epsilon}+1}e^{-\epsilon |x - \mu|}
    \end{split}
\end{equation}
where $\mu = 1-\frac{1}{\epsilon} \ln(\delta (e^{\epsilon} + 1))$ (resp. $\mu = \frac{1}{\epsilon} \ln(\delta (e^{\epsilon} + 1))-1$).}\vspace{-.5mm}
\end{definition}  

{Next, we detail the generation of (noisy) join key MFs (Alg~\ref{algo:synopgen} lines 6:10). We assume that both $A_{\mathsf{ft}}$ and $A_{\mathsf{j}}$ are non-empty and that \(\mathbf{h}(\mathsf{pair}_k, D)\) partitions \(dom(A_{\mathsf{ft}})\) into bins $\{B_1, \dots, B_m\}$. Each owner then generates a table of noisy MFs, \smash{\(\{\widehat{\mathtt{mf}}(A_{\mathsf{j}}, D^{\ell})\}_{1 \leq \ell \leq m}\)}, where each entry \smash{\(\widehat{\mathtt{mf}}(A_{\mathsf{j}}, D^{\ell})\)} represents an independently generated MF statistic for the join key attribute \(A_{\mathsf{j}}\), calculated over a specific subset of data filtered by the attribute \(A_{\mathsf{ft}}\), such that}
\begin{equation}\label{eq:mf}
\widehat{\mathtt{mf}}(A_{\mathsf{j}}, D^{\ell}) \gets \widehat{\max}_{\epsilon}\left(\mathcal{G}_{\mathbf{count}(A_{\mathsf{j}})}\left(\sigma_{A_{\mathsf{ft}}\in B_{\ell}}(D)\right)\right)
\end{equation}  
here $\mathcal{G}_{\mathbf{count}(A_{\mathsf{j}})}$ is a group-by-count operation over $A_{\mathsf{j}}$, and $\widehat{\max}_{\epsilon}$ is a {\it report noisy max} mechanism~\cite{dwork2014algorithmic}. It first adds i.i.d. noise from the exponential distribution $\mathsf{Exp}(\frac{2}{\epsilon})$ to each grouped count, then outputs the largest noisy count. We stress that \(\mathcal{M}_{\synop}\) will not generate noisy MFs for non-join key attributes, and when \(A_{\mathsf{ft}}\) is empty, a global MF will be generated instead of MF tables (Alg~\ref{algo:synopgen}:10). Moreover, since \sys enables owners to label attributes as unique-valued, if \(A_{\mathsf{j}}\) is known to be unique-valued, then \(\smash{\widehat{\mathtt{mf}}(A_{\mathsf{j}}, \cdot)}\) is always 1. Nevertheless, as exponential noises are non-negative, thus \smash{$\widehat{\mathtt{mf}} \geq \mathtt{mf}$} holds for all cases. 

\begin{theorem}\label{tm:priv} Given $|\mathsf{pair}| = c$, $\epsilon, \delta>0$, the synopsis generation (Algorithm~\ref{algo:synopgen}) is $(\hat{\epsilon}, \hat{\delta})$-DP where $\hat{\epsilon} \leq 6\epsilon\sqrt{c\ln(1/\delta)}$, and $\hat{\delta} = (c+1) \delta$
\end{theorem}
For space concern, we move complete proofs to the full version~\cite{special2024online}. In a sketch, adding $\mathsf{Lap}^{+}$ (or $\mathsf{Lap}^{-}$) to a single bin is $(\epsilon, \delta)$-DP. By parallel and sequential composition, generating $\mathbf{{H}}$ is $(2\epsilon, 2\delta)$-DP. Moreover, each noisy max is $(\epsilon, 0)$-DP, and by parallel composition, the generation of the entire MF table is also $(\epsilon, 0)$-DP. In this way, we know that the generation of each $(\mathsf{pair}_k, \mathbf{H}, \mathtt{MF}_k)$ is at most $(3\epsilon, 2\delta)$-DP. Given there are in total $c$ such pairs, and thus the total privacy loss is subject to $c$-fold advanced composition~\cite{dwork2014algorithmic}.

\vspace{3pt}\noindent{\bf Synopsis transformations.} We say that one can perform transformations on released synopses without incurring extra privacy loss, per the post-processing theorem of DP~\cite{dwork2014algorithmic}. Now, we outline the key synopsis transformation relevant to \sys's design. First, given any (2d) bounding histogram \(\mathbf{H}(\mathsf{pair}, D)\) with both \(A_{\mathsf{ft}}, A_{\mathsf{j}} \in \mathsf{pair}\) are non-empty, one can derive the (1d) bounding histograms, i.e. \(\mathbf{H}(A_{\mathsf{j}}, D)\) and \(\mathbf{H}(A_{\mathsf{ft}}, D)\), for any single attribute \(A_{\mathsf{ft}}\) or \(A_{\mathsf{j}}\) by marginal sums $\mathbf{h}^{+}, \mathbf{h}^{-} \in \mathbf{H}(\mathsf{pair}, D)$ over $A_{\mathsf{j}}$ or $A_{\mathsf{ft}}$, respectively. This enables the creation of statistics on individual attributes, even when \(A_{\mathsf{ft}}\) and \(A_{\mathsf{j}}\) are not included as a standalone synopsis attribute. Moreover, it's possible to derive relevant join key statistics following a selection on the base relation. For example, given $A_{\mathsf{ft}}, A_{\mathsf{j}} \in \mathsf{pair}\neq \emptyset$, and let $D'\gets\sigma_{A_{\mathsf{ft}} \in vals}(D)$, one can obtain the (1d) bounding histogram $\mathbf{H}(A_{\mathsf{j}}, D')$ by conducting a selective marginal sum of $\mathbf{h}^{+}, \mathbf{h}^{-} \in \mathbf{H}(\mathsf{pair}, D)$ over bins of $A_{\mathsf{ft}}$ that intersect with \(vals\). Beyond bounding histograms, join key MFs over pre-filtered data can also be computed by  
\begin{equation}~\label{eq:mf}
    \widehat{\mathtt{mf}}(A_{\mathsf{j}}, D') = \min\left(\textstyle \sum_{B_{\ell} \cap vals \neq \emptyset} \widehat{\mathtt{mf}}(A_{\mathsf{j}}, D^{\ell}),  \widehat{\mathtt{mf}}(A_{\mathsf{j}}, D)\right)
\end{equation}
Note that \(\widehat{\mathtt{mf}}(A_{\mathsf{j}}, D)\) exists if \(A_{\mathsf{j}}\) is also included as a standalone synopsis attribute; otherwise, Eq~\ref{eq:mf} yields only the first term.

\section{\sys Primitives}\label{sec:primitive}
We next introduce the secure primitives in \sys. 
One major challenge in designing these primitives is the knowledge gap on how private synopses can accelerate oblivious query processing, i.e., \hyperref[challenge-c3]{C-3}. To address this, we explore various usage of synopses, including creating private indexes (\pidx~\S~\ref{subsec:dpindex}) and designing compacted oblivious operations (\pop~\S~\ref{sec:eops}). Given that joins are the most resource-intensive operations, we optimize join algorithms by combining private indexing and compaction techniques to develop a novel, parallel-friendly oblivious join (\S~\ref{sec:eops}). Another challenge is ensuring lossless processing, which we tackle by integrating mechanisms that pessimistically estimate selection cardinalities, indexing ranges, and join sizes using synopses (\S~\ref{sec:prim}) and advanced upper bound techniques~\cite{hertzschuch2021simplicity}. For simplicity, we assume all input relations are of size \( n \) and all 1D histograms contain \( m \) bins.

\vspace{1mm}

\subsection{Basic Operations}\label{sec:pops}
\sys supports conventional fully-oblivious operators~\cite{bater2017smcql}, which logically the same as non-private ones but with data-independent execution and worst-case padding for results. We briefly introduce these operations: (i) {\bf Default data access (\code{SeqACC}).} By default, query execution begins with loading all data into a secure array via sequential scan. Each loaded tuple gets a secret bit \(\mathsf{ret}\) (initially `0'), marking its validity; (ii) {\bf \code{SELECT}.} This secure filter \(\sigma_p(R)\) performs a linear scan over the secure array \( \langle R\rangle \), updating the \(\mathsf{ret}\) bit to `1' for tuples satisfying predicate \( p \) and `0' for others; (iii) {\bf\code{PROJECT}.} Removes irrelevant attributes from relation \( \langle R\rangle \), but retains the \(\mathsf{ret}\) bit; (iv) {\bf \code{JOIN}.} Implements a secure \( \theta \)-join \(R_0 \Join_{\theta} R_1\) by computing the cartesian product \( \langle R_0 \times R_1\rangle \) and marking joined tuples with \code{SELECT}. The output is padded to the worst-case maximum size; (v) {\bf \code{COUNT}, \code{SUM}, \code{MIN/MAX}.} These aggregation operators scan the secure array and update a secret-shared aggregation value for each tuple; (vi) {\bf\code{ORDER-BY}, \code{DISTINCT}, \code{GROUP-BY(AGG)}.} Built on the oblivious sort primitive~\cite{batcher1968sorting}. \code{ORDER-BY} sorts the array by a given attribute. \code{DISTINCT} sorts and then identifies unique tuples, marking only the last in a sequence of identical tuples with \(\mathsf{ret} = \text{`1'}\). \code{GROUP-BY(AGG)} first uses \code{DISTINCT} to find unique tuples. For each distinct tuple (\( \mathsf{ret} \) set to `1'), it appends an aggregation value derived from the tuple and a dummy attribute (e.g., `-1') for non-distinct tuples.

\subsection{\sys Index \pidx}\label{subsec:dpindex}
Existing index techniques in SCA have several drawbacks: loss of qualified data~\cite{roy2020crypt}, reliance on intricate data structures with large overhead~\cite{zhang2023longshot, bogatov2021epsolute, qiudoquet}, and restricted query support~\cite{zhang2023longshot, bogatov2021epsolute}. Furthermore, all these techniques only support indexing the base relations. \pidx offers a breakthrough by enabling the creation of lossless indexes directly on outsourced data and the query of intermediate results, eliminating the need for extra structures or storing dummy data.





In general, \pidx builds upon the typical indexing model that utilizes cumulative frequencies (CF)~\cite{kraska2018case}. Specifically, given \(D\) sorted by \(A\in attr(D)\), all records \(t \in D\) where \(t.A = x\) can be indexed by the interval \([g(x-1), g(x)]\), where \(g(x) = |\{t \mid t.A \leq x\}|\) is the CF function. For better illustration, we show an example index lookup in Figure~\ref{fig:index}: to get all records with an attribute value of 24, one may compute \([g(23), g(24)] = [217, 248]\) and access the relevant data from the subset \(D[217:248]\). To make this indexing method private and lossless, the key idea of \pidx is to derive two noisy CF curves from synopses (i.e., bounding histograms). One curve, \(g^{+}(x)\), consistently overestimates \(g(x)\), while the other, \(g^{-}(x)\), consistently underestimates it. Then for any attribute value \(x\), we can now derive a private interval \([g^{-}(x-1), g^{+}(x)]\) that losslessly indexes all the desired records. As illustrated in Figure~\ref{fig:index}, the \sys index might estimate the index range for attribute value 24 as \([g^{-}(23), g^{+}(24)] = [198, 267]\).





\begin{figure}[h]
\centering
\includegraphics[width=0.95\linewidth,interpolate=false]{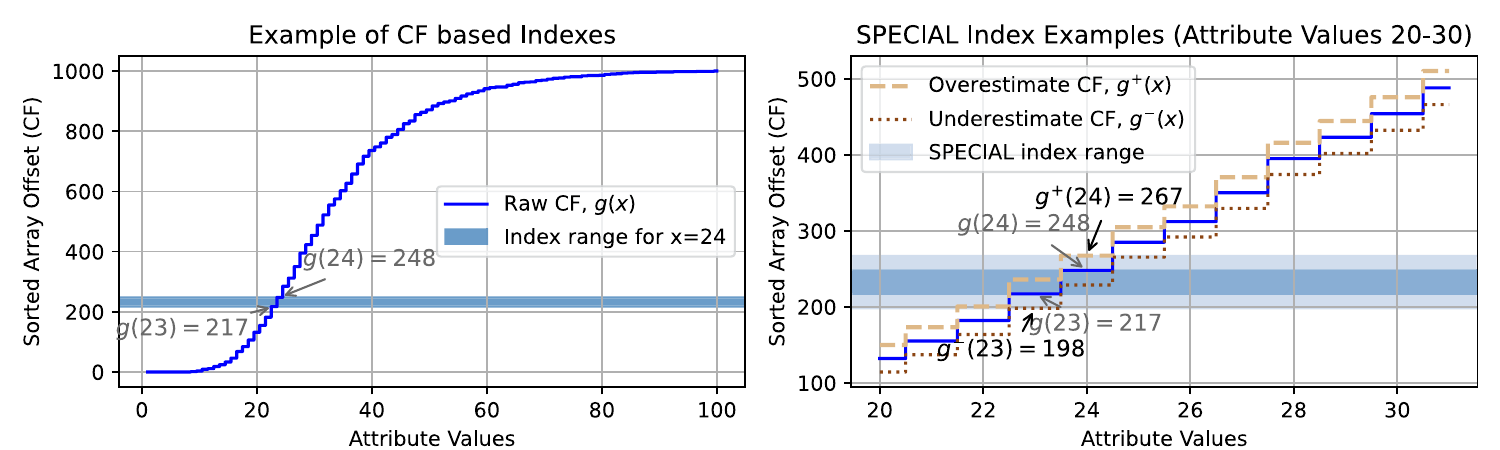}
\caption{True (left) vs. \sys (right) index for $x=24$}
\label{fig:index}
\end{figure}

In what follows, we provide the formal explanations on how \pidx derives indexes from DP synopses. Specifically, \pidx first determines the bounding histograms $\mathbf{H}(A, D)$, which may be either transformed from an available 2D histogram $\mathbf{H}(\mathsf{pair}, D)$ with $A\in \mathsf{pair}$, or sourced directly if $\mathbf{H}(A, D)$ is already included in the synopses. It then constructs the noisy mapping as follows:
\begin{definition}[\sys index] \emph{Given $D$ sorted by $A$, the bounding histogram $\mathbf{H}(A, D) = \{\mathbf{h}^{+}, \mathbf{h}^{-}\}$, and assume $\mathbf{h}^{+}=(c^{+}_1, ...,c^{+}_m)$, \(\mathbf{h}^{-}=(c^{-}_1, ...,c^{-}_m)\) partitions $dom(A)$ into $\{B_1,...,B_m\}$. We say $\pidx(A, D) = \{\mathsf{idx}_i=[\mathsf{lo}_i, \mathsf{hi}_i]\}_{ 1\leq i\leq m}$ is the \sys index of $D$ over $A$ with: 
\begin{itemize}
    \item $\forall~i\geq 1$, $\mathsf{hi}_i = \min(|D|, \sum^{i}_{k=1} c^{+}_k)$.
    \item $\mathsf{lo}_1 =0$, and $\forall~i\geq 2$, $\mathsf{lo}_i = \sum^{i-1}_{k=1} \max(0, c^{-}_k)$.
\end{itemize} 
}
\end{definition}

By this construction, all tuples $t \in D$ such that $t.A \in B_i$ will be organized into the subset $D[\mathsf{idx}_i] \subseteq D$. This subset can be quickly accessed if $D$ is already sorted, without the need for special data structures or inclusion of dummy tuples. Depending on how bounding histograms are constructed, $\pidx(A, D)$ can support indexing lookups with varying granularity. This can range from indexing individual attribute values (where each $B_i$ corresponds to a single domain value) to indexing a range of of values. 
The bounding histogram's pessimistic estimation ensures that all tuples where \(t.A \in B_i\) are accurately contained within \(D[\mathsf{idx}_i]\), thereby achieving lossless indexing. In contrast to existing methods that are limited to indexing base relations \cite{roy2020crypt, bogatov2021epsolute, zhang2023longshot}, \pidx extends its capabilities to create private indexes on query intermediate results. For instance, consider \(D' \gets \sigma_{A^*\in vals}(D)\) where the attribute pair \((A^*, A)\) is included in \(\synop\). Here, \pidx can derive \(\mathbf{H}(A, D')\) from \(\mathbf{H}(A^*, D)\) and subsequently build indexes on \(D'\). Importantly, since index creation is a post-processing procedure using available DP synopses, it incurs no additional privacy loss.





\vspace{3pt}\noindent{\bf Indexed store and fast data access \code{IdxAcc}.} \pidx enables a new storage layout for outsourced data, namely indexed datastore. Specifically, by analyzing a representative workload $Q_{\mathsf{R}}$, one may identify the ``hottest'' attribute per base relation, sort them according to the ``hottest'' attribute, and then build indexes over the sorted data. This storage layout enables fast indexed access (\code{IdxAcc}) to retrieve a compact subset of data from the outsourced relations, thereby eliminating the need for a full table sequential scan (\code{SeqAcc}) and can directly produce a compact input. We emphasize that the \sys design does not require replicating the outsourced datastore to accommodate multiple query types~\cite{bogatov2021epsolute, zhang2023longshot}. However, creating compact replicas (e.g., column replicas~\cite{huang2020tidb} over frequently queried attributes) can be optionally employed to enhance query processing speed. Moreover, the generation of all aforementioned objects (indexed store and column replicas) requires only three primitives: projection, oblivious sorting, and \pidx. In other words, this implies that one can selectively adjust these objects to align with dynamic query workloads, without incurring extra privacy loss.



\subsection{\sys Operators \pop}\label{sec:eops}
We introduce \pop, a set of novel synopsis-assisted operators that maintain full obliviousness, 
while enabling lossless compaction. 
To our knowledge, \pop is the first primitive of its kind in any SCA.




\vspace{3pt}\noindent{\bf Oblivious compaction: \textbf{\code{OPAC}}}. is a fundamental operation critical to other \pop primitives. Given input \( \langle R \rangle \), \code{OPAC} sorts it based on the secret bit \( \mathsf{ret} \), moving tuples with \( \mathsf{ret} = \) `1' to the front. Then, \code{OPAC} retains only the first $k$ tuples from the sorted array. The compaction is {\it lossless} if \( k \) is greater than or equal to the number of tuples with \( \mathsf{ret} = \) `1'; otherwise, it is {\it lossy}.

\vspace{3pt}\noindent{\bf \sys selections: \code{(OP)SELECT}, \code{(SP)SELECT}, \code{(DC)SELECT}.} Let $R$ to be a relation and $A\in attr(R)$, we now introduce three advanced selections that implements \(\sigma_{A\in vals}(R)\). 

\noindent\underline{\code{(OP)SELECT}.} is mainly implemented based on the oblivious compaction (\code{OPAC}) operation. Specifically, the operation first conducts a standard \code{SELECT} on the input secure array \( \langle R \rangle \) to label selected tuples, followed by an \code{OPAC} to to eliminate a large portion of non-matching tuples. To determine the compaction size $cs$, \code{(OP)SELECT} examines the synopsis of \(R\) and pessimistically estimates the cardinality of \(\sigma_{A\in vals}(R)\) as shown in Algorithm~\ref{algo:cardest}. Since $cs$ never underestimates the actual cardinality, and thus, the compaction is lossless with no missing tuples. Moreover, as \code{OPAC} is fully oblivious and \( cs \) is determined completely from post-processing over DP synopsis, thus, \code{(OP)SELECT} causes no privacy loss. 

\begin{algorithm}[]
\caption{$\ce(\sigma_{A\in vals}(R), \synop)$}
\begin{algorithmic}[1]
\State $\mathbf{h} = \emptyset$, $c=0$
\If{$\mathbf{H}(A, R) \in \synop$}~$\mathbf{h} \gets \mathbf{h}^{+} \in \mathbf{H}(A, R)$
\ElsIf{$\exists~ \mathsf{pair} \in \synop$, s.t. $A \in \mathsf{pair}$}
\State $\mathbf{h} \gets$ marginal sum $\mathbf{h}^{+} \in \mathbf{H}(A, R)$ over $\left(\mathsf{pair} \setminus A\right)$.
\Else ~{\bf return} $cs=|R|$
\EndIf
\State {\bf return} $cs = \min(|R|, \sum_{i=1}^{m}c_i \in \mathbf{h} : \left(B_m \cap vals \neq \emptyset\right))$
\end{algorithmic}
\label{algo:cardest}
\end{algorithm}

\vspace{1mm}
\noindent\underline{\code{(SP)SELECT}.} The running complexity of \code{(OP)SELECT} depends on \code{OPAC}, which is typically linearithmic (see \S~\ref{subsec:cm} or~\cite{sasy2022fast}). However, when $\ce(\sigma_{A\in vals}(R), \synop)$ is relatively small, oblivious selection can be achieved without necessarily incurring linearithmic cost. Specifically, we consider \code{(SP)SELECT}, which first creates an empty output array $\langle R_o \rangle$ with size equals to $cs$ before any computations. Next, it evaluates two linear scans over $\langle R \rangle$, where the first scan obliviously marks all selected tuples, and in the second scan, it privately writes all marked tuples into $\langle R_o \rangle$. Specifically, in the second scan, \code{(SP)SELECT} internally maintains the last {\it actual write} position $\mathsf{idx}$ in $\langle R_o \rangle$. Then for every newly accessed tuple \( \langle t\rangle \) in \( \langle R \rangle \), a write action occurs on all tuples in \( \langle R_o \rangle \). If \( \langle t\rangle \) is selected, then an {\it actual write} is made that writes \( \langle t\rangle \) to $\langle R_o[\mathsf{idx}+1] \rangle$ and a {\it dummy write} is made to elsewhere. If not, dummy writes are made throughout \( \langle R_o \rangle \). We say that, in the context of the secret-shared secure array $\langle\mathbf{a}\rangle$, a dummy write to $\langle\mathbf{a}[i]\rangle$ is simply a re-sharing of $\mathbf{a}[i]$ through secure protocols without changing its value. 

\vspace{1mm}
\noindent\underline{\code{(DC)SELECT}.} Finally, if the underlying data is already indexable on \(A\), a direct pre-fetch can be applied to avoid full table scan and compaction. The operator simply looks up \(\pidx(A, R)\), and accesses $R[a, b]$, where $a=\min_i(\mathsf{idx}_i.\mathsf{lo}), 
b= \max_i(\mathsf{idx}_i.\mathsf{hi})$, and $\mathsf{idx}_i$ dentoes the index in \(\pidx(A, R)\) with bin $B_i \cap [a,b] \neq \emptyset$. A standard \code{SELECT} is then applied to $R[a, b]$. 




\vspace{3pt}\noindent{\bf \sys join: \code{(MX)JOIN}.} We now introduce a novel MF-Index based oblivious join operation. The advancements of \code{(MX)JOIN} stand out in three aspects. First, compared to the standard \code{JOIN}, \code{(MX)JOIN} stands out for its ability to significantly compact the output size, coupled with a highly parallelizable fast processing mode. Second, existing DP oblivious joins typically require spending privacy budget~\cite{dong2022r2t} to learn join sensitivity~\cite{dong2022r2t} or necessitate truncation on joined tuples~\cite{wang2022incshrink, bater2018shrinkwrap}. \code{(MX)JOIN} eliminates this need. Moreover, \code{(MX)JOIN} is unique as the first oblivious join that enables lossless output compaction without extra privacy loss. We illustrate the construction details in Algorithm~\ref{algo:mxjoin}.

\begin{algorithm}[ht]
\caption{$\code{(MX)JOIN}$ (base and pre-filtered relations)}
\begin{algorithmic}[1]
\Statex \textbf{Input}: relations $R_0$, $R_1$; join attribute $A_{\mathsf{j}}$; we consider synopses (histograms) of $A_{\mathsf{j}}$ are partitioned into bins $B_1, ...B_m$.
\If{$\mathsf{MXReady}(R_0, R_1) == \mathsf{True}$  }~ $\mathsf{BucketJoin}(R_0, R_1, A_{\mathsf{j}})$
\ElsIf{$R_0$, $R_1$ are either base or pre-filtered relation}
\For{$b\in\{0,1\}$}
\State derive $\widehat{\mathtt{mf}}(A_{\mathsf{j}}, R_b)$ from $\synop_b$ (\S~\ref{sec:prim})
\State build index $\pidx(A_{\mathsf{j}}, R_b) = \{\mathsf{idx}_i\}_{i=1,..,m}$ (\S~\ref{subsec:dpindex})
\EndFor
\If{$\forall b~$, $\widehat{\mathtt{mf}}(A_{\mathsf{j}}, R_b)$, and $\pidx(A_{\mathsf{j}}, R_b) \neq \mathsf{null}$ }
\State oblivious sort $R_0$, $R_1$ on $A_{\mathsf{j}}$, {$\mathsf{BucketJoin}(R_0, R_1, A_{\mathsf{j}})$}
\Else ~{\bf assert} ``not applicable for \code{(MX)JOIN}''
\EndIf
\EndIf
\Statex \underline{$\mathsf{BucketJoin}(R_0, R_1, A_{\mathsf{j}})$}:
\For{$i = 1,2,...,m$}
\State \smash{$R^{(i)}_{0,1}\gets \sigma_{A_j \in B_i}(R_{0,1})$ using \code{(DC)SELECT}/ $\pidx(A_{\mathsf{j}}, R_{0,1})$}
\State compute  \smash{$O_i\gets (R^{(i)}_0 \Join_{A_{\mathsf{j}}} R^{(i)}_1)$} via standard \code{JOIN}
\State $cs_i \gets \min\left(\frac{|R^{(i)}_0|}{\widehat{\mathtt{mf}}(A_{\mathsf{j}}, R_0)}, \frac{|R^{(i)}_1|}{\widehat{\mathtt{mf}}(A_{\mathsf{j}}, R_1)}\right)\times {\widehat{\mathtt{mf}}(A_{\mathsf{j}}, R_0)}\cdot{\widehat{\mathtt{mf}}(A_{\mathsf{j}}, R_1)}$
\State $R_{\mathsf{out}}\gets R_{\mathsf{out}} \cup \code{OPAC}(O_i, cs_i)$
\EndFor
\State {\bf return} $R_{\mathsf{out}}$
\end{algorithmic}
\label{algo:mxjoin}
\end{algorithm} 
In general, \code{(MX)JOIN} can be applied to two types of data: the base and pre-filtered relations where the join key attribute is included in $\synop$. Specifically, \code{(MX)JOIN} starts with computing the join key MFs (Alg~\ref{algo:mxjoin}:4) and constructing private indexes (Alg~\ref{algo:mxjoin}:5) for both inputs. All these operations are conducted through ``privacy cost-free'' transformations using available DP synopses. Once these objects are obtained, the algorithm employs oblivious sort to rearrange both inputs (Alg~\ref{algo:mxjoin}:6,7), rendering them indexable with tuples logically distributed into independent buckets by join key values. Next, \code{(MX)JOIN} simply adopts standard \code{JOIN} to join tuples exclusively within the same buckets (Alg~\ref{algo:mxjoin}:10). Finally, \code{(MX)JOIN} performs per-bucket output compaction, where it first determines the \emph{MF join bound}~\cite{hertzschuch2021simplicity} for each bucket join and invokes \code{OPAC} to compact the output according to the learned size (Alg~\ref{algo:mxjoin}:11,12). As bucket-wise operations are independent, the aforementioned steps lend themselves well to parallelized processing. As \code{(MX)JOIN} derives join compaction sizes completely from post-processing of DP synopses, it thus incurs no extra privacy loss. Additionally, the noisy MF bounds guarantee that compaction sizes are consistently overestimated, ensuring lossless compaction of join results.

\section{\sys Planner}\label{sec:planner}
Current DPSCA designs struggle with costly execution plans because they cannot pre-estimate query intermediate sizes, and thus unable to identify effective execution plans with minimized cost. \sys overcomes this challenge by introducing a novel query planner that uses synopses for size estimation, and thus enabling both private and efficient SCA query planning. 

At a high level, our planner is modeled after the Selinger-style optimizer~\cite{blasgen1981system}. It uses a bottom-up, dynamic programming approach to enumerate all equivalent plans for a given query, estimates their costs (heavily influenced by intermediate sizes) using available synopses, and selects the plan with minimal cost. The introduction of \sys primitives significantly impacts cost modeling for oblivious operations, rendering existing models~\cite{bater2018shrinkwrap, liagouris2023secrecy} inadequate. Furthermore, the design-space challenge of query planning persists, and the extensive equivalent plan search space necessitates strategies to simplify the search process. To address these, we first systematically analyze the complexities of \sys primitives and develop a new cost model (\S~\ref{subsec:cm}). We then design protocol-specific heuristics (\S~\ref{subsec:huristic}) tailored to our planner to narrow the search space.

\subsection{Cost Model} \label{subsec:cm}
\re{We adopt the standard SCA cost framework~\cite{bater2018shrinkwrap} to develop \sys's cost model, viewing the cost of a secure execution plan as the sum of each operator's I/O and secure evaluation costs.} Specifically, given a plan with $\ell$ operators, $\mathsf{op}_1,...,\mathsf{op}_{\ell}$, and let $\mathbf{I}=\{{I}_1,..., {I}_{\ell}\}$, $\mathbf{O}=\{{O}_1,..., {O}_{\ell}\}$, to be the input and output sizes of each operators. The plan cost is:
\begin{equation}
\mathsf{Cost} = \textstyle \sum_{i=1}^{\ell} C^{\mathsf{op}_i}_{\mathsf{in}}(I_i) + C_{\mathsf{eval}}^{\mathsf{op}_i}(I_i) + C^{\mathsf{op}_i}_{\mathsf{out}}(O_i)
\end{equation}
Here, \(C_{\mathsf{in}}\) represents the data access cost (input I/O), primarily capturing the expenses when moving data from persistent storage to an in-memory secure array. 
\(C_{\mathsf{out}}\) denotes the output I/O cost, modeling the expenses when writing operator results into output arrays.  
 \(C_{\mathsf{eval}}\) accounts for the secure computing cost for evaluating an operator, typically constituting the dominant cost. Note that, in practice, the exact formulas for \(C_{\mathsf{in}}, C_{\mathsf{out}}\), and \(C_{\mathsf{eval}}\) can vary depending on the specific secure protocol employed (garbled circuits~\cite{yao1986generate}, secret sharing~\cite{micali1987play}, etc.) as well as the particular hardware configurations in use. Nonetheless, the understanding of the asymptotic costs is adequate for comprehending the principles of SCA query planning and optimization strategies~\cite{liagouris2023secrecy, bater2018shrinkwrap}. In what follows, we provide detailed analysis on the asymptotic costs for each \sys operator. Similarly, we assume that all input data sizes mentioned henceforth in this section are of size \(n\), and all 1D histograms have \(m\) bins.  Table~\ref{tab:asymptotic-costs} summarizes the operator costs.

\vspace{1mm}
\begin{table}[]
\centering
\caption{Asymptotic costs for secure operators}\vspace{-3mm}
\label{tab:asymptotic-costs}
\scalebox{0.77}{
\begin{tabular}{|c c c c|}
\hline
\textbf{Operator} & \textbf{Input I/O ($C_{\mathsf{in}}$)} & \textbf{Eval. ($C_{\mathsf{eval}}$)} & \textbf{Output I/O ($C_{\mathsf{out}}$)} \\
\hline
\code{PROJECT} & $O(n)$ & $\text{N/A}$ & $O(n)$ \\
Agg. & $O(n)$ & $O(n)$ & $O(1)$ \\
Group \& Order & $O(n)$ & $O(n\log^2n)$ & $O(n)$ \\
\hline
\code{SELECT} & $O(n)$ & $O(n)$ & $O(n)$ \\
\code{(OP)SELECT} & $O(n)$ & $O(n \log n)$ & $\mathsf{hist\_bound}$ \\
\code{(SP)SELECT} & $O(n)$ & $O(n)$ & $O(1)$ \\
\code{(DC)SELECT} & $\mathsf{idx\_bound}$ & N/A & N/A \\
\hline
\code{JOIN} & $O(n)$ & $O(n^2)$ & $O(n^2)$ \\
\code{(MX)JOIN} & $O(n)$ & $O(n^2)^*$ & $\mathsf{mf\_bound}$ \\
\hline 
\end{tabular}}
\noindent{$^*$\footnotesize{Assuming the max size of the indexed buckets is bounded by $O(\frac{n}{\log n})$.}}
\vspace{-1em}
\end{table}
\vspace{2pt}\noindent{\bf Oblivious sorting and compaction.} While oblivious sorting algorithms with optimal $O(n\log n)$ complexity exist, they often necessitate either impractically large constants~\cite{ajtai19830, goodrich2014zig} or  client-side memory~\cite{asharov2020bucket}, both do not fit with SCA scenario. Consequently, we will consider the well-established bitonic sorting based implementation for oblivious sort, which come with  $O(n\log^2 n)$ complexity. Nonetheless, efficient \code{OPAC} implementations with $O(n\log n)$ complexity remain achievable~\cite{sasy2022fast}.

\vspace{2pt}\noindent{\bf Projection, grouping and aggregation.} The \code{PROJECT} accesses private relations and discards unnecessary columns independently on each server, which is naturally oblivious. Thus, I/O costs dominate this operation, with both input and output costs bounded by \(O(n)\). The costs of \code{ORDER-BY}, \code{DISTINCT}, and \code{GROUP-BY} are dominated by oblivious sorting, resulting in a complexity of \(O(n\log^2n)\). Additionally, as these operators do not reduce output sizes, both \(C_{\mathsf{in}}\) and \(C_{\mathsf{out}}\) are bounded by \(O(n)\). Finally, the cost of aggregations, i.e. \code{COUNT}, \code{SUM}, and \code{MIN/MAX} subjects to a oblivious linear scan, typically outputting a single secret-shared value. Hence, its \(C_{\mathsf{eval}}\) is bounded by \(O(n)\), with $C_{\mathsf{in}}$ at \(O(n)\) and $C_{\mathsf{out}}$ at \(O(1)\).

\vspace{2pt}\noindent{\bf Selections.} The primary cost of \code{SELECT} stems from an oblivious linear scan, making \(C_{\mathsf{eval}}\) within \(O(n)\). Since \code{SELECT} does not shrink the output size, both \(C_{\mathsf{in}}\) and \(C_{\mathsf{out}}\) are within \(O(n)\). \code{(OP)SELECT} requires an oblivious compaction (\code{OPAC}) before writing outputs, where \code{OPAC} usually yields an \(O(n\log n)\) complexity~\cite{sasy2022fast}. Consequently, its \(C_{\mathsf{eval}}\) is bounded by \(O(n\log n)\) with input I/O cost same as \code{SELECT}. However, as \code{(OP)SELECT} compacts output size, \(C_{\mathsf{out}}\) is reduced to \(\mathsf{hist\_bound} = O\left(\sum_{B_i \cap vals \neq \emptyset} |B_i|\right)\), where \(B_i \cap vals \neq \emptyset\) are bins in the synopsis histogram intersecting with selection conditions. If \(\sum_{B_i \cap vals \neq \emptyset} |B_i| \sim O(1)\), \code{(SP)SELECT} becomes preferable, with its running cost dominated by a two-phase linear scan, and thus  \(C_{\mathsf{eval}}\) is now \(O(n)\), and the output cost is \(O(1)\). \code{(DC)SELECT} is the most efficient selection, though it requires indexable input data. All costs are directly related to the size of the indexed data, so \(C_{\mathsf{in}}\), \(C_{\mathsf{evals}}\) and \(C_{\mathsf{out}}\) are all bounded by \(\mathsf{idx\_bound} = O\left(\max_i(\mathsf{idx}_i.\mathsf{hi}) - \min_i(\mathsf{idx}_i.\mathsf{lo})\right)\). Here, $\mathsf{idx}_i$ are indexed regions that intersect with selection conditions.


\vspace{3pt}\noindent{\bf Joins.} Both \code{JOIN} and \code{(MX)JOIN} have \(O(n)\) data access costs, but differ in \(C_{\mathsf{eval}}\) and \(C_{\mathsf{out}}\). \code{JOIN}, conducting a Cartesian product for two input tables, has \(C_{\mathsf{eval}}\) and \(C_{\mathsf{out}}\) both bounded by \(O(n^2)\).
Compared to \code{JOIN}, in the worst-case scenario where the join keys follow a highly biased distribution, i.e. max bucket size reaches \(O(n)\), \code{(MX)JOIN}'s asymptotic cost is at most \(O(n^2\log n)\). However, when join keys are distributed more uniformly, the cost can be asymptotically better. For instance, with \(m = \log n\) and assuming a max bucket size of \smash{\(O(\frac{n}{\log n})\)}, each bucket join costs \smash{\(O(\frac{n^2}{\log n})\)}, leading to a total cost of \(O(n^2)\), equivalent to \code{JOIN}. Recall that bucket joins in \code{(MX)JOIN} can be executed concurrently, hence, the processing latency is indeed dominated by the bucket-wise cost, i.e. \smash{\(O(\frac{n^2}{\log n})\)}. Additionally, the output cost is lowered from \(O(n^2)\) to the sum of per-bucket MF upper bounds (Alg~\ref{algo:mxjoin}:11), which can be substantially less if the join key MFs are low.

\subsection{Heuristics}\label{subsec:huristic}
\noindent{\bf H-1. Filter push down.} is a common query planning optimization, moves selection operations to the earliest possible stage to reduce data processed by subsequent operations. In conventional SCAs, data obliviousness often requires padding selection sizes, making filter pushdown ineffective~\cite{liagouris2023secrecy,eskandarian2017oblidb, zheng2017opaque, poddar2021senate}. However, \sys's innovative selection methods enable compacting selections to approximate true cardinalities without compromising privacy, restoring the effectiveness of filter pushdown. Hence, we include filter push down as one of the optimization heuristic for \planner.


\vspace{2pt}\noindent{\bf H-2. Predicates fusion.} Let $R$ to be any relation, $A_1, A_2, ..., A_k \subseteq attr(R)$, and $\mathbf{v} = \{v_1, v_2,...,v_k\}$. We say that for multiple selection over $R$ such that $\sigma_{A_1 \in v_1} (...\sigma_{A_k\in v_k}(R))$, one can always fuse them into one selection $\sigma_{\mathbf{A}\subset \mathbf{v}}(R)$. This can reduce the number of secure computation invocations from $k$ rounds to just one. Additionally, the selection size can be estimated as $\min_{A_i}\left(\ce(\sigma_{A_i\in v_i}(R))\right)$.  

\vspace{2pt}\noindent{\bf H-3. Join statistics propagation.} A key property of \sys join, \code{(MX)JOIN}, is that the output is already indexed and bucketized by the join key. Therefore, for any output $R$ of \code{(MX)JOIN} computing $R_0 \Join_{A_{\mathsf{j}}} R_1$, a new index $\pidx(A_{\mathsf{j}}, R)$ across $R$ can be easily derived. Moreover, as per~\cite{hertzschuch2021simplicity},  one can also update the MF for $R$ by computing $\widehat{\mathtt{mf}}(A_{\mathsf{j}}, R) = \widehat{\mathtt{mf}}(A_{\mathsf{j}}, R_0) \times \widehat{\mathtt{mf}}(A_{\mathsf{j}}, R_1)$. As a result, we say that the output of \code{(MX)JOIN} as \emph{MF-and-index-ready}, enabling direct application of another \code{(MX)JOIN} on the same join attribute.

\section{Evaluation}\label{sec:eva}
In this section, we present evaluation results of our proposed framework. Specifically, we address the following questions: \textbf{Question-1}: Does \sys offer efficiency and privacy advantages over existing SCAs?  \textbf{Question-2}: For \sys design, is there a trade-off between privacy guarantees and system efficiency? \re{How different synopses scenarios affect system performance?} \textbf{Question-3}: Can \sys scale complex analytical (e.g. multi-way join) queries to large-scale (multi-million rows) datasets?

\subsection{Experimental setups} \label{sec:setup}
\noindent{\bf Baseline systems and \sys prototype.} We compare \sys with two baseline systems: Shrinkwrap~\cite{bater2018shrinkwrap}, the SOTA DPSCA, and SMCQL~\cite{bater2017smcql} (also used as a baseline for Shrinkwrap). For consistency, we consider the same circuit-model implementations for both baseline systems and the \sys prototype. While some works~\cite{poddar2021senate,liagouris2023secrecy} similar to SMCQL use exhaustive padding but improve efficiency through protocol-level optimizations, we exclude them from our benchmarks for fair comparison concerns. 
We re-implemented key query features for the two baseline systems and built \sys using the same MPC package ( \code{EMPtoolkit-0.2.5}). 
All implementations are open-sourced~\cite{special2024}.

\vspace{4pt}\noindent{\bf Datasets and workloads.} 
\re{ We developed two open benchmarks. The first reproduces the HealthLNK benchmark used by Shrinkwrap and SMCQL which simulates a real-world scenario where medical researchers want to perform secure analytics across multiple cohorts' sensitive data. 
We use an open schema~\cite{smcql-git} to generate a scalable synthetic dataset with 4 tables, 222K rows and 24 columns. The benchmark involves four multi-cohort medical study queries, identical to those used in Shrinkwrap and SMCQL. 
Our second benchmark simulates secure collaborative analytics within the financial sector. Imagine multiple banks needing to analyze their combined, private data to study loan and financial statistics—all without compromising sensitive customer information.  We use the anonymized {\em Czech Financial Dataset}~\cite{RelationalDataFinancial} for this, assuming each entry represents data owned by a different bank or financial organization unable to directly share information.} This dataset comprises 8 relational tables with a total of 55 columns and 1.1 million rows. For testing workloads, we design eight query workloads, ranging from simple linear queries to complex multi-way join-aggregation queries. A brief summary of the workloads is provided in Table~\ref{tab:workloads}.



\vspace{4pt}\noindent{\bf Default configurations.} For \sys, we employ a \re{per-table privacy budget allocation strategy. Each table is allocated a one-time privacy budget of $\epsilon=1.5$ and $\delta=0.00005$ for synopses generation, and is evenly distributed across all DP synopses.} For Shrinkwrap, we adopt their default \re{per-query privacy allocation, assigning a privacy budget of $\epsilon=1.5$ and $\delta=0.00005$ to each query, as outlined in~\cite{bater2018shrinkwrap}. It is important to note that this configuration means Shrinkwrap will not offer guarantees on the bounded privacy loss across multiple queries.} For the HealthLNK benchmark, we use {\em Dosage} and {\em Aspirin} as representative workloads, and for the Financial benchmark, we use {\em FQ2, FQ4, and FQ8.} 
Unless further elaborated, these workloads will also serve as default testing queries for our evaluation. For all equal-width histograms generated in \sys, we configure them to have at most 8 bins. Moreover, for baseline systems, as they do not have join ordering optimizations, thus we will assume a random join order for them. We conduct all experiments on bare-metal Mac machines with M2 Max CPUs and 96GB unified memory.  

\begin{table}[]
\caption{\re{Query workloads}}
\label{tab:workloads}
\scalebox{0.7}{
\begin{tabular}{|c|c|c|l|}
\hline
\textbf{Bench.}                     & \textbf{Query}          & \textbf{Type}                    & \multicolumn{1}{c|}{\textbf{Description}} \\ \hline
\multirow{4}{*}{\textbf{HealthLNK}} & \textbf{Dosage Study}   & \multicolumn{1}{l|}{Binary Join} & Expanding binary join.                    \\
                                    & \textbf{Comorbidity}    & \multicolumn{1}{l|}{Binary Join} & Non-expanding binary join.                \\
                                    & \textbf{Aspirin count}  & \multicolumn{1}{l|}{Multi Join}  & 3 way mixed join                          \\
                                    & \textbf{3 Join Aspirin} & \multicolumn{1}{l|}{Multi Join}  & 4 way mixed join                          \\ \hline
\multirow{8}{*}{\textbf{Financial}} & \textbf{FQ1}            & Linear                           & Point query.                              \\
                                    & \textbf{FQ2}            & Linear                           & Range query.                              \\
                                    & \textbf{FQ3}            & Binary Join                      & Non-expanding binary join.                \\
                                    & \textbf{FQ4}            & Binary Join                      & Expanding binary join.                    \\
                                    & \textbf{FQ5}            & Multi Join                       & 3-way mixed joins.                        \\
                                    & \textbf{FQ6}            & Multi Join                       & 3-way all expanding joins.                \\
                                    & \textbf{FQ7}            & Multi Join                       & 4-way mixed joins.                        \\
                                    & \textbf{FQ8}            & Multi Join                       & 5-way mixed joins.                        \\ \hline
\end{tabular}}
\end{table}

\begin{figure*}[ht]
\centering
\includegraphics[width=0.92\linewidth,interpolate=false]{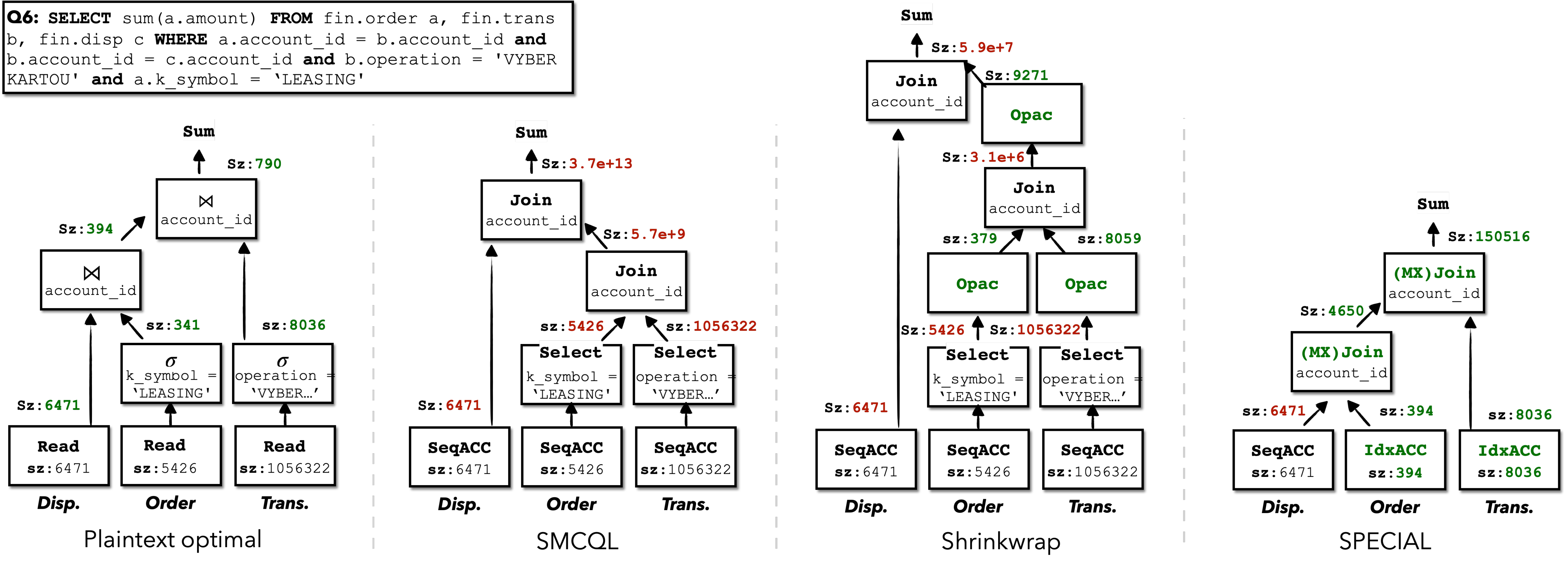}
\caption{In-depth comparisons in execution plans: (i) The exhaustive padding in SMCQL can lead to significant memory blowup; (ii) Both SMCQL and Shrinkwrap suffer from \re{unoptimized join ordering}; (iii) Although Shrinkwrap reduces intermediate sizes, it still requires substantial memory to materialize join outputs; (iv) \sys can identify efficient join execution orders to reduce intermediate sizes; (v) \code{IdxAcc} can significantly reduce input I/O costs.}
\label{fig:plan_cmp}
\end{figure*}

\subsection{End-to-end comparisons}\label{sec:ete}
To address {\bf Question-1}, we first conduct an end-to-end performance comparison of \sys, Shrinkwrap, and SMCQL across all benchmark workloads. The results are summarized in Figure~\ref{fig:all:time},~\ref{fig:all:mem}. We cannot complete full benchmark for SMCQL due to high memory cost, so we project evaluations for FQ4 (using 10\% data) and omit results for other complex workloads.
\begin{figure}[H]
\includegraphics[width=0.45\textwidth,interpolate=false]{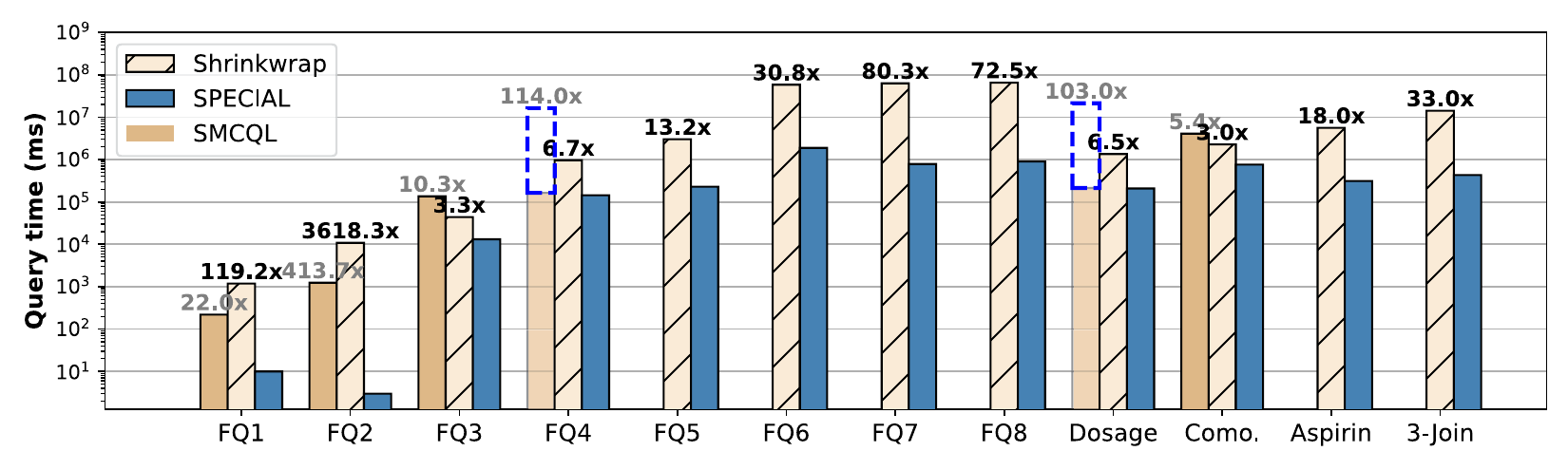}
\vspace{-1em}
\caption{\re{End-to-end comparison: query latency}}
\label{fig:all:time}
\end{figure}
\vspace{-2em}
\begin{figure}[H]
\includegraphics[width=0.45\textwidth,interpolate=false]{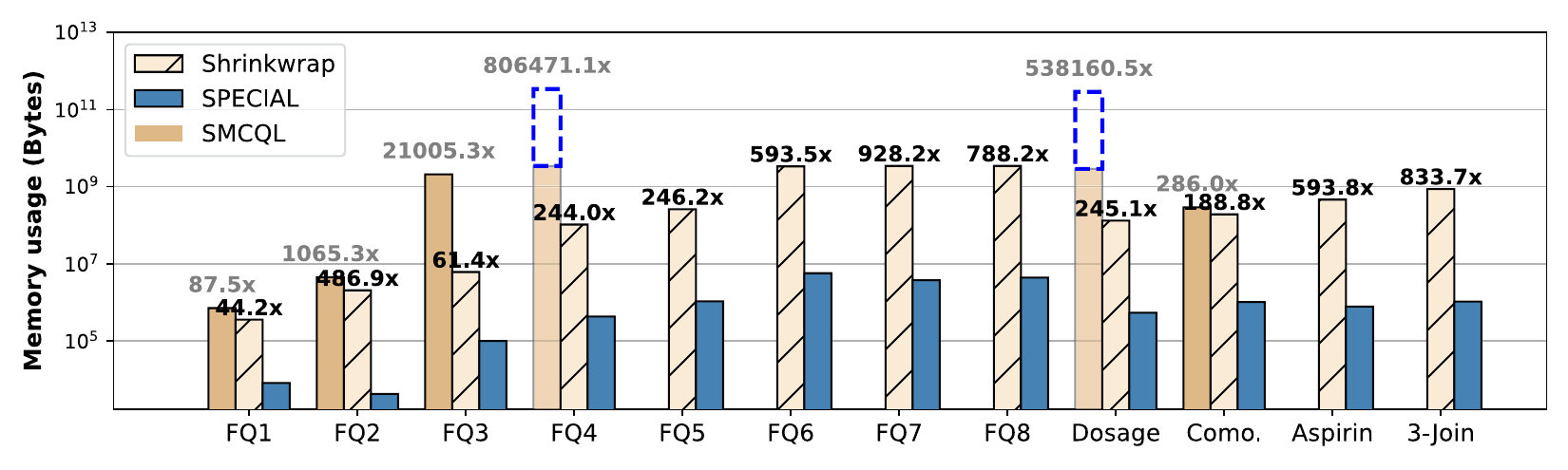}
\vspace{-1em}
\caption{\re{End-to-end comparison: memory usage}}
\label{fig:all:mem}
\end{figure}

\noindent{\bf Observation 1. \re{\sys outperforms Shrinkwrap and SMCQL in query latency across all benchmarks,} reaching up to $3618.3\times$ for linear queries, $114\times$ for binary joins, and $80.3\times$ for multi-joins.} Figure~\ref{fig:all:time} shows the comparison results in query time. \re{First, \sys shows significant speedups for linear queries, reaching up to $3618.3\times$ (FQ2). This large performance gain is mainly attributed to its index-based fast data access. By directly fetching private data through DP indexes, \sys eliminate substantial I/O costs (e.g., sequential reads) and bypass the need for secure computations (e.g., oblivious filter).} Second, we observe that in binary joins, \sys has a less pronounced advantage over Shrinkwrap. This is because binary joins have a single join order, eliminating the potential for join ordering problems (where different join orders lead to significantly different performance). Consequently, even though Shrinkwrap does not optimize join orders, it doesn't experience efficiency losses in this scenario. However, for more complex multi-way joins, \sys's advantage becomes more pronounced again. For instance, more than $80\times$ speedup in FQ7. This is because \sys can pre-select efficient join orders before runtime. 

\vspace{3pt}\noindent{\bf Observation 2. \sys shows profound improvement in memory usage (Figure~\ref{fig:all:mem}) against baseline systems, especially in complex multi-way joins.} This is primarily due to two factors: First, the \code{(MX)JOIN} used by \sys is more memory-efficient compared to the joins implemented by Shrinkwrap and SMCQL. Second, \sys's capability to identify optimal execution plans significantly reduces total intermediate sizes, which is particularly beneficial for complex joins that suffer from sub-optimal or exhaustive padding in other systems. To better understand the substantial improvements \sys achieves—for instance, up to $928.2\times$ over Shrinkwrap and more than $10^5\times$ over SMCQL—we will zoom into a specific query, FQ6, and compare the detailed execution plans of the three systems. The choice of FQ6 is strategic because its complexity sufficiently highlights the differences in execution plans, yet it remains simple enough for clear visual representation. Our comparison features four execution plans: a plaintext optimal plan, illustrating the ground truth optimal execution; a hypothetical SMCQL plan (with projected  cardinalities); and two actual execution plans from our experiments with Shrinkwrap and \sys. The detailed comparisons and observations are summarized in Figure~\ref{fig:plan_cmp}.


To continue address {\bf Question-1}, we now compare the privacy guarantees of \sys with the baseline systems (Figure~\ref{fig:privacy}). Specifically, we focus on comparing cumulative privacy loss in multiple query answering, w.r.t. two composition models: advanced composition (Adv.)\cite{dwork2014algorithmic} and concentrated composition (CDP.)\cite{bun2016concentrated}.

\begin{figure}[H]
\captionsetup[sub]{font=small,labelfont={bf,sf}}
    \begin{subfigure}[b]{0.5\linewidth}
\centering\includegraphics[width=1\linewidth,interpolate=false]{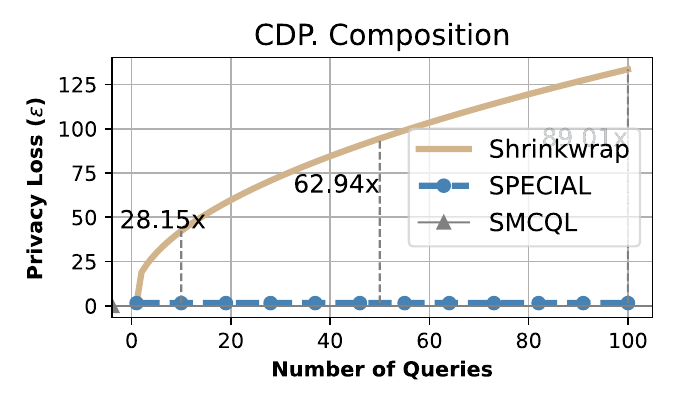}
    \end{subfigure}
    \begin{subfigure}[b]{0.5\linewidth}
\centering\includegraphics[width=1\linewidth,interpolate=false]{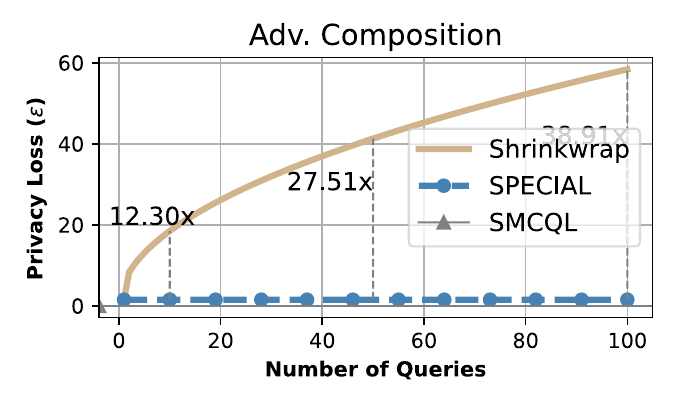}  
    \end{subfigure}
   \caption{\re{End-to-end comparison privacy loss.}}
   \label{fig:privacy}
\end{figure}

\noindent{\bf Observation 3. Under continual query answering, \sys demonstrates significantly lower privacy loss compared to Shrinkwrap, achieving up to $89.01\times$ and $38.91\times$ improvements in the Adv. and CDP modes, respectively.} The privacy loss of \sys is bounded to the initial synopsis release stage, so continual query answering does not incur additional privacy loss. In contrast, Shrinkwrap's privacy loss accumulates over time as each new query allocates a fresh privacy budget. Consequently, its privacy loss exhibits a logarithmic growth, as shown in Figure~\ref{fig:privacy}. This accumulation can result in significant privacy degradation when processing a large number of queries. For example, answering 100 queries in Shrinkwrap could result in a privacy loss of $\epsilon>100$ in Adv. and $\epsilon\approx 60$ in CDP., respecitvely, even if each query only uses a small privacy budget of $\epsilon=1.5$. As such, \sys demonstrate significant improve in privacy guarantees towards SOTA DPSCA. \re{Even when compared to standard SCA (e.g., SMCQL) with no privacy loss due to exhaustive padding, our system incurs only a small and fixed privacy cost (e.g. $\epsilon=1.5$ per table) while delivering substantial performance gains.} 


\subsection{Privacy efficiency tradeoffs}\label{subsec:tradeoff}
We address {\bf Question-2} by evaluating \sys at various privacy levels. Specifically, we maintain $\delta$ constant while varying $\epsilon$ from $0.1$ to $10$ and assess the performance across default testing queries. 
The results are shown in Figure~\ref{fig:eps2}.
\begin{figure}[H]
\includegraphics[width=0.48\textwidth,interpolate=false]{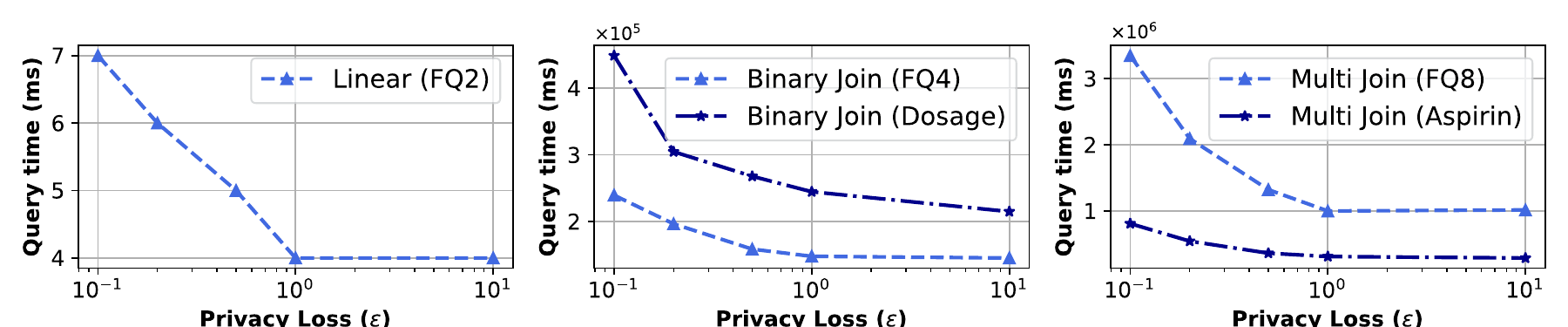}
\vspace{-3em}
\end{figure}
\begin{figure}[H]
\includegraphics[width=0.48\textwidth,interpolate=false]{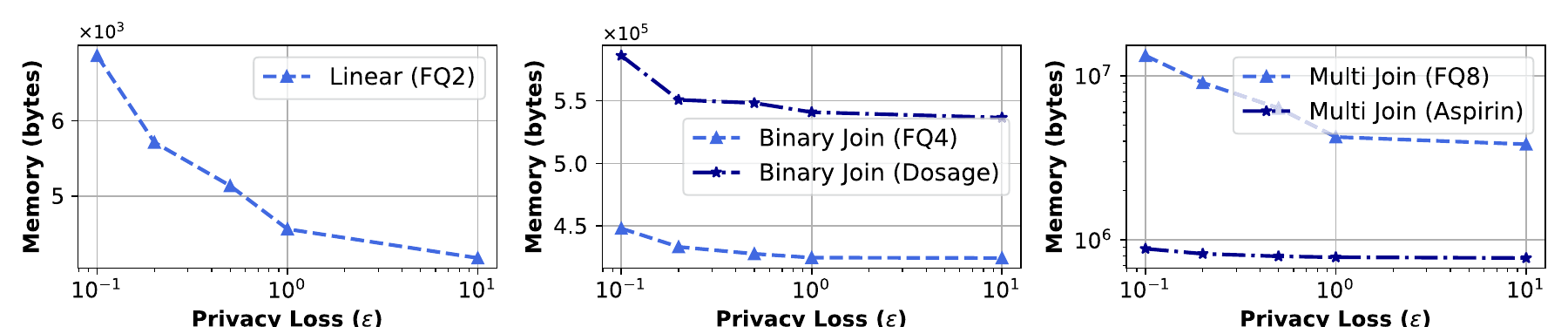}
\caption{\re{Privacy vs. Performance trade-offs}}
\label{fig:eps2}
\end{figure}
\vspace{3pt}\noindent{\bf Observation 4. The privacy-efficiency tradeoff generally exists but exhibits varying trends at different privacy levels. For instance, \sys shows a clear tradeoff at higher privacy levels ($\epsilon < 1$), while at lower privacy levels ($\epsilon > 1$), the tradeoff becomes less pronounced.} When $\epsilon$ increases from $0.1$ to $1$, both memory usage and query latency for all test queries significantly decrease. However, increasing $\epsilon$ from $1$ to $10$ shows no significant performance gains. This may indicate that once $\epsilon$ exceeds $1$, the impact of noises on cardinality estimation or index building is alredy minimal, and further reductions in $\epsilon$ do not lead to notable improvements. Therefore, if high privacy protection is required, practitioners should carefully fine-tune privacy parameters to optimize performance. Conversely, if performance is the priority, setting $\epsilon$ near 1 is typically sufficient.

\re{
\subsection{Synopses impacts micro benchmarks}\label{subsec:micro}
We continue to address {\bf Question-2} to explore how synopses scenarios may affect \sys's performance. Specifically, we study two key settings: (i) How different bin numbers (BinNum) in synopses can impact the efficiency of \code{(IDX)JOIN}, and (ii) how synopsis coverage levels for a single query can affect its overall execution. We will conduct micro-benchmarks for a thorough investigation. Note that simulate different synopses scenario on both HealthLNK and Financial benchmarks can be challenging (e.g., joins typically occur on the same key, so it is hard to simulate partial coverage), hence, to better control experimental variables and accurately assess impacts, we will now use synthetic data and workloads.

\begin{figure}[ht]
\captionsetup[sub]{font=small,labelfont={bf,sf}}
    \begin{subfigure}[b]{0.49\linewidth}
    \centering\includegraphics[width=1\linewidth,interpolate=false]{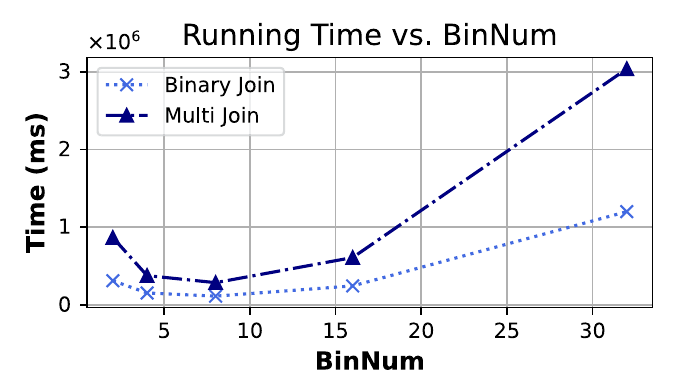}
    \end{subfigure}
    \begin{subfigure}[b]{0.49\linewidth}
    \centering\includegraphics[width=1\linewidth,interpolate=false]{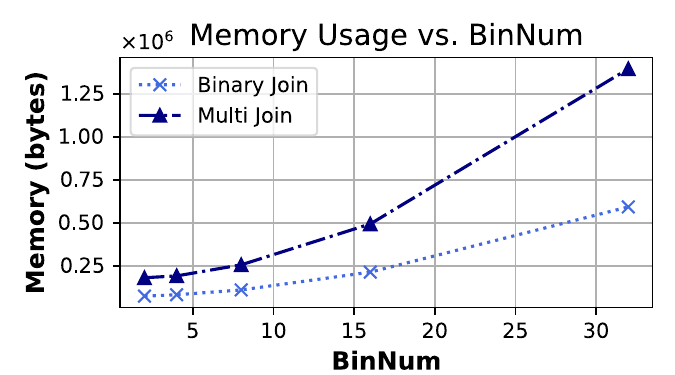}  
    \end{subfigure}
   \caption{\re{BinNum experiments.}}
   \label{fig:binnum}
\end{figure}

We first study how the BinNums impact the performance of \code{(IDX)JOIN}. To study this, we synthesize two join queries on fixed input data, generate join key synopses with varying bin numbers (2 to 64), and measure the performance of \code{(IDX)JOIN} in processing the queries. The results are shown in Figure~\ref{fig:binnum}.

\vspace{3pt}\noindent{\bf Observation 6. The running time of \code{(IDX)JOIN} initially decreases but then increases as the BinNum grows. The memory usage consistently increases.} \code{(IDX)JOIN} partitions larger joins into smaller sub-joins, but since we rely on DP indexes for partitioning, each sub-join inevitably includes additional dummy data. This increased plan size directly translates to higher memory usage and will grow when BinNum increases (more noises in DP indexes). On the other hand, partitioning large joins into smaller ones can enhance join efficiency, which is why we initially observe a decrease in execution time as BinNums increase (e.g., 2 to 8). However, the trade-off arises when the number of bins becomes excessive (e.g., $>8$). The overhead of handling the increased dummy data starts to outweigh the benefits gained from partitioning. At this point, the performance improvement plateaus and starts to degrade as the system struggles with the inflated plan sizes. 

\begin{table}[]
\caption{\re{Synopses coverage experiments.}}
\label{tab:synop}
\scalebox{0.7}{
\begin{tabular}{cccccc}
\cline{2-6}
\multicolumn{1}{c|}{\multirow{2}{*}{\textbf{}}}                                                          & \multicolumn{1}{c|}{\textbf{Synopses Coverage}} & \multicolumn{1}{c|}{\textbf{Time (ms)}} & \multicolumn{1}{c|}{\textbf{Improv.}} & \multicolumn{1}{c|}{\textbf{Mem. (bytes)}} & \multicolumn{1}{c|}{\textbf{Improv.}} \\ \cline{2-6} 
\multicolumn{1}{c|}{}                                                                                    & \multicolumn{1}{c|}{\textbf{No coveragee}}      & \multicolumn{1}{c|}{1831703}            & \multicolumn{1}{c|}{baseline}                 & \multicolumn{1}{c|}{129048576}             & \multicolumn{1}{c|}{baseline}                 \\ \hline
\multicolumn{1}{|c|}{\multirow{3}{*}{\textbf{\begin{tabular}[c]{@{}c@{}}Join Key \\ (JK)\end{tabular}}}} & \multicolumn{1}{c|}{\textbf{2 way (1 JK)}}      & \multicolumn{1}{c|}{276340}             & \multicolumn{1}{c|}{6.6$\times$}               & \multicolumn{1}{c|}{28481424}              & \multicolumn{1}{c|}{4.5$\times$}               \\
\multicolumn{1}{|c|}{}                                                                                   & \multicolumn{1}{c|}{\textbf{3 way (2 JKs)}}     & \multicolumn{1}{c|}{29200}              & \multicolumn{1}{c|}{62.7$\times$}              & \multicolumn{1}{c|}{905808}                & \multicolumn{1}{c|}{142.5$\times$}             \\
\multicolumn{1}{|c|}{}                                                                                   & \multicolumn{1}{c|}{\textbf{All way (3 JKs)}}   & \multicolumn{1}{c|}{6097}               & \multicolumn{1}{c|}{300.4$\times$}             & \multicolumn{1}{c|}{65088}                 & \multicolumn{1}{c|}{1982.7$\times$}            \\ \hline
\multicolumn{1}{|c|}{\multirow{4}{*}{\textbf{Filter$^*$}}}                                                   & \multicolumn{1}{c|}{\textbf{1 input}}           & \multicolumn{1}{c|}{444463}             & \multicolumn{1}{c|}{4.1$\times$}               & \multicolumn{1}{c|}{40327680}              & \multicolumn{1}{c|}{3.2$\times$}               \\
\multicolumn{1}{|c|}{}                                                                                   & \multicolumn{1}{c|}{\textbf{2 inputs}}          & \multicolumn{1}{c|}{108205}             & \multicolumn{1}{c|}{16.9$\times$}              & \multicolumn{1}{c|}{12602400}              & \multicolumn{1}{c|}{10.2$\times$}              \\
\multicolumn{1}{|c|}{}                                                                                   & \multicolumn{1}{c|}{\textbf{3 inputs}}          & \multicolumn{1}{c|}{24486}              & \multicolumn{1}{c|}{74.8$\times$}              & \multicolumn{1}{c|}{3943200}               & \multicolumn{1}{c|}{32.72$\times$}             \\
\multicolumn{1}{|c|}{}                                                                                   & \multicolumn{1}{c|}{\textbf{All inputs}}        & \multicolumn{1}{c|}{6847}               & \multicolumn{1}{c|}{267.5$\times$}             & \multicolumn{1}{c|}{1303200}               & \multicolumn{1}{c|}{99$\times$}                \\ \hline
\multicolumn{6}{l}{* We synthesize a random selectivity between (0, 0.33) for each filter operation.}  
\end{tabular}}
\end{table}
We now study how synopses coverage impacts query performance. We synthesize and test a 4-way join query under two controlled scenarios: (i) {\em JK coverage:} We focus on varying the level of JK coverage, starting from 1 out of 3 JKs to full coverage, while ensuring no filter synopses are present. We then measure how this impacts query performance; and (ii) {\em Filter coverage}: We maintain full JK coverage and change the coverage of filter synopses on the query's input tables, ranging from 1 out of 4 inputs to full coverage. This allows us to examine the isolated effect of filter synopsis coverage on performance. Results are in Table~\ref{tab:synop}. 

\vspace{3pt}\noindent{\bf Observation 7. For both groups, query efficiency significantly improves as synopses coverage grows. Nevertheless, even at the lowest coverage level, queries can still achieve notable speedups.} Even with minimal synopsis coverage—like boosting only one join or applying synopses to just one input table—we observe significant speedups of $6.6\times$ and $4.1\times$, respectively. This demonstrates the potential for substantial performance gains even with limited synopsis availability. Moreover, real-world workloads often involve joins on the same keys and similar filtered inputs (e.g., HealthLNK workloads), suggesting that high synopsis coverage is achievable even with a small set of representative workloads.  As demonstrated by Table~\ref{tab:synop}, adding even a single additional synopsis, whether for join keys or table filters, can yield a substantial performance boost (up to 10$\times$) in query execution efficiency.
}

\subsection{Scaling experiments}\label{subsec:scale-exp}
To address {\bf Question-3}, we stress \sys with two types of scaled workloads: (i) {\em scaled data:} we duplicate the raw dataset to sizes of $2\times$, $4\times$, and $8\times$ and evaluate default testing queries; \re{{\em (ii) scaled query complexity:} We use standard inputs, but simulate complex multi joins (up to 9-way) by chaining together multiple join workloads.} The results are shown in Figure~\ref{fig:scale} and Table~\ref{tab:jscale}.
\begin{figure}[ht]
\captionsetup[sub]{font=small,labelfont={bf,sf}}
    \begin{subfigure}[b]{0.49\linewidth}
    \centering\includegraphics[width=1\linewidth,interpolate=false]{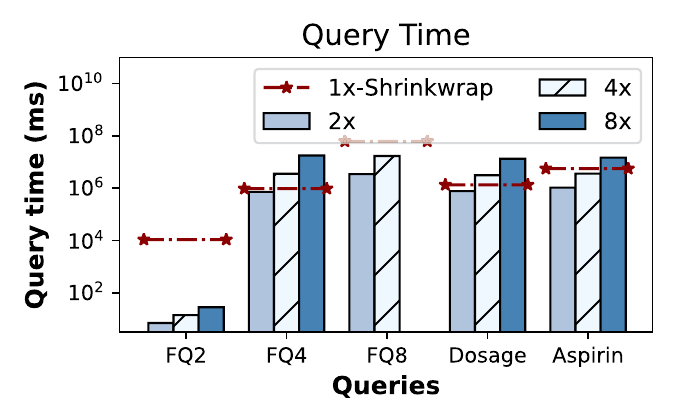}
    \end{subfigure}
    \begin{subfigure}[b]{0.49\linewidth}
    \centering\includegraphics[width=1\linewidth,interpolate=false]{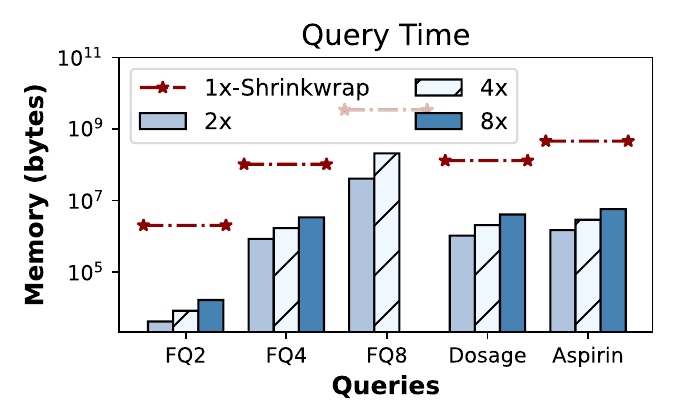}  
    \end{subfigure}
   \caption{\re{Scaling data experiments.}}
   \label{fig:scale}
\end{figure}

\vspace{2pt}\noindent{\bf Observation 8. \sys shows large potential to scale up to multi-million data, even for complex 5-way joins.} Figure~\ref{fig:scale} shows \sys's effective scaling: up to $8\times$ data for linear queries and binary joins, and up to $4\times$ data for complex 5-way joins like Q8. For instance, Q2 can be completed within 290ms under $8\times$, and in fact, since selection is bypassed (due to index access), thus the cost is mainly on I/O costs. Q4 finishes in 289 minutes at the same scale $8\times$, while the more complex 5-way join Q8 takes less than 280 minutes for $4\times$ data. As a reference, Shrinkwrap would require over 1035 minutes to complete Q8 even with unscaled data. 

\re{\vspace{3pt}\noindent{\bf Observation 9. \sys can effectively process very large joins (e.g., 9-ways).} Table~\ref{tab:jscale} shows that the query processing time of \sys at 9-way scale can still be $4.4\times$ faster than Shrinkwrap at 3-way join scale (Q6), and the memory improvement is even more evident that is 328$\times$ smaller. We stress that these significant memory savings can become even more crucial when processing massive datasets. Techniques like Shrinkwrap or SMCQL, might be forced to rely on much slower persistent storage to handle intermediate results that exceed memory capacity. In contrast, \sys can still maintain a fully in-memory query mode, potentially leading to even more pronounced efficiency gains in such scenarios.
}

\begin{table}[]
\caption{\re{Scaling query complexity experiments}}
\label{tab:jscale}
\scalebox{0.72}{
\begin{tabular}{|c|c|c|c|c|}
\hline
\textbf{Join scale}                      & \textbf{Query time (ms)} & \multicolumn{1}{l|}{\bf Improv.}  & \textbf{Memory (bytes)} & \multicolumn{1}{l|}{\textbf{Improv.}} \\ \hline
\textbf{Shrinkwrap Q6}       & 58342782                 & \multicolumn{1}{l|}{baseline} & 3354165360              & \multicolumn{1}{l|}{baseline}         \\ \hline
\textbf{7 (FQ8 $\Join$ FQ3)}             & 1330602                  & 43.8$\times$                          & 5334360                 & 628$\times$                                   \\
\textbf{8 (FQ8 $\Join$ FQ5)}             & 11789800                 & 4.9$\times$                           & 7901016                 & 424$\times$                                   \\
\textbf{9 (FQ8 $\Join$ FQ3 $\Join$ FQ4)} & 13306150                 & 4.4$\times$                           & 8999352                 & 328$\times$                                   \\ \hline
\end{tabular}}
\end{table}

\section{Related Work}\label{sec:related}
\noindent{\bf SCA systems.} Two main approaches exist for designing MPC-based SCA systems.  The first is {\em peer-to-peer (P2P) paradigm}~\cite{aggarwal2005two, mouris2024delegated, poddar2021senate, volgushev2019conclave, wang2021secure}, where the goal is to improve efficiency by decomposing analytical queries and pushing them to data owners, so that they can either directly process in clear or running MPCs across a small group of parties. Unfortunately, this approach imposes large overhead on data owners, especially for complex operations like joins. Given that real-world data owners often lack robust computing resources and  service capabilities, the P2P paradigm is hard to scale and support reliable SCA services to external analysts. The other paradigm is the {\em server-aided-MPC model}~\cite{wang2021dp, bater2018shrinkwrap, bater2020saqe, wang2022incshrink, zhang2023longshot, roy2020crypt, kamara2011outsourcing, tan2021cryptgpu, knott2021crypten, liagouris2023secrecy, mohassel2017secureml}. This model allows data owners to outsource both expensive MPC computations and secure data storage to a set of capable servers, which can then jointly evaluate MPC to provide reliable SCA services. \sys is built upon the server-aided-MPC model and under a strong ``all but one''corruption assumption. Moreover, \sys's core design is protocol-agnostic, which allows interoperability with various MPC models, including the P2P or a weaker corruption where a supermajority of servers need to be honest~\cite{tan2021cryptgpu, knott2021crypten, liagouris2023secrecy}.


\vspace{3pt}\noindent{\bf DP leakages.} Leakage-abuse attacks~\cite{cash2015leakage, kellaris2016generic, blackstone2019revisiting, oya2021hiding, grubbs2018pump, shang2021obfuscated, zhang2016all}, exploit data-dependent processing patterns, are persistent threats to SCA systems. To mitigate these risks, oblivious computation~\cite{keller2014efficient, pinkas2010oblivious, stefanov2011towards, chang2022towards, mishra2018oblix, crooks2018obladi, jonsson2011secure, sasy2022fast, asharov2020bucket, wang2014oblivious, pettai2015combining, wang2023private, tinoco2023enigmap} have become the {\em de facto} solution.
While this technique ensures the strongest privacy guarantees by eliminating any data-dependent leakages, it also introduces a fundamental contention with modern database optimizations, which often rely heavily on data-dependent operations~\cite{bater2018shrinkwrap,wang2021dp,wagh2021dp,wang2022incshrink}. To this end, many recent efforts seek a practical balance in the privacy-performance trade-off by allowing controlled leakage under DP~\cite{bater2018shrinkwrap,wang2021dp,wagh2021dp,wang2022incshrink, groce2019cheaper, he2017composing, zhang2023longshot, qin2022adore, chan2022foundations, wagh2016differentially, chu2021differentially, qiudoquet, patel2019mitigating}. However, a common issue of these approaches is unbounded privacy loss. While some works propose to address this~\cite{zhang2023longshot, bogatov2021epsolute, qiudoquet}, their approaches are restricted to only simple linear queries. 
\sys addresses all these limitations together, and to our knowledge, is the first SCA system that can simultaneously ensure both bounded privacy and lossless results for complex SPJA queries.



\vspace{4pt}\noindent{\bf SCA query planning.} Query planning~\cite{silberschatz2011database} is crucial in conventional databases. Conventional planners can exploit size disparities across different query plans to choose efficient ones with smaller sizes~\cite{elmasri2020fundamentals, harmouch2017cardinality, han2021cardinality}. However, such techniques use data-dependent information and are typically prohibited in SCA. A handful of studies~\cite{bater2017smcql, volgushev2019conclave, poddar2021senate, liagouris2023secrecy} that explore query planning within SCA frameworks primarily rely on data-independent metrics for planning, which usually lead to only moderate optimizations. Shrinkwrap~\cite{bater2018shrinkwrap} introduced a private planning method that optimally compacts intermediate sizes by efficiently allocating privacy budgets to minimize dummy data. However, it cannot pre-determine an optimal join order. \sys offers an advanced query planner capable of pre-estimating intermediate sizes and comparing execution costs among different plan structures before runtime.

\section{Conclusion}
We introduce \sys, the first SCA system that simultaneously supports: (i) handling complex queries with bounded privacy loss; (ii) advanced query planning that effectively exploit plan intermediate sizes before runtime; and (iii) delivering exact query results without missing tuples. This is achieved through a novel synopses-assisted SCA design, where a set of private table statistics are released with one-time privacy cost to guide subsequent secure SCA planning and processing. 

\begin{acks}
 This work was supported by the NSF grant OAC-2419821, IUIAS Collaborative Award, and IU Faculty Startup Grant. 
\end{acks}

\balance
\bibliographystyle{ACM-Reference-Format}
\bibliography{sample}


\begin{thebibliography}{80}


\ifx \showCODEN    \undefined \def \showCODEN     #1{\unskip}     \fi
\ifx \showDOI      \undefined \def \showDOI       #1{#1}\fi
\ifx \showISBNx    \undefined \def \showISBNx     #1{\unskip}     \fi
\ifx \showISBNxiii \undefined \def \showISBNxiii  #1{\unskip}     \fi
\ifx \showISSN     \undefined \def \showISSN      #1{\unskip}     \fi
\ifx \showLCCN     \undefined \def \showLCCN      #1{\unskip}     \fi
\ifx \shownote     \undefined \def \shownote      #1{#1}          \fi
\ifx \showarticletitle \undefined \def \showarticletitle #1{#1}   \fi
\ifx \showURL      \undefined \def \showURL       {\relax}        \fi
\providecommand\bibfield[2]{#2}
\providecommand\bibinfo[2]{#2}
\providecommand\natexlab[1]{#1}
\providecommand\showeprint[2][]{arXiv:#2}

\bibitem[\protect\citeauthoryear{??}{Rel}{[n.d.]}]%
        {RelationalDataFinancial}
 \bibinfo{year}{[n.d.]}\natexlab{}.
\newblock \bibinfo{title}{Financial Dataset}.
\newblock \bibinfo{howpublished}{\url{https://relational-data.org/dataset/Financial}}.
\newblock
\newblock
\shownote{Accessed: 2024-03-30.}


\bibitem[\protect\citeauthoryear{Aggarwal, Bawa, Ganesan, Garcia-Molina, Kenthapadi, Motwani, Srivastava, Thomas, and Xu}{Aggarwal et~al\mbox{.}}{2005}]%
        {aggarwal2005two}
\bibfield{author}{\bibinfo{person}{Gagan Aggarwal}, \bibinfo{person}{Mayank Bawa}, \bibinfo{person}{Prasanna Ganesan}, \bibinfo{person}{Hector Garcia-Molina}, \bibinfo{person}{Krishnaram Kenthapadi}, \bibinfo{person}{Rajeev Motwani}, \bibinfo{person}{Utkarsh Srivastava}, \bibinfo{person}{Dilys Thomas}, {and} \bibinfo{person}{Ying Xu}.} \bibinfo{year}{2005}\natexlab{}.
\newblock \showarticletitle{Two Can Keep A Secret: A Distributed Architecture for Secure Database Services.}. In \bibinfo{booktitle}{\emph{CIDR}}, Vol.~\bibinfo{volume}{2005}. \bibinfo{pages}{186--199}.
\newblock


\bibitem[\protect\citeauthoryear{Ajtai, Koml{\'o}s, and Szemer{\'e}di}{Ajtai et~al\mbox{.}}{1983}]%
        {ajtai19830}
\bibfield{author}{\bibinfo{person}{Mikl{\'o}s Ajtai}, \bibinfo{person}{J{\'a}nos Koml{\'o}s}, {and} \bibinfo{person}{Endre Szemer{\'e}di}.} \bibinfo{year}{1983}\natexlab{}.
\newblock \showarticletitle{An 0 (n log n) sorting network}. In \bibinfo{booktitle}{\emph{Proceedings of the fifteenth annual ACM symposium on Theory of computing}}. \bibinfo{pages}{1--9}.
\newblock


\bibitem[\protect\citeauthoryear{Araki, Furukawa, Lindell, Nof, and Ohara}{Araki et~al\mbox{.}}{2016}]%
        {araki2016high}
\bibfield{author}{\bibinfo{person}{Toshinori Araki}, \bibinfo{person}{Jun Furukawa}, \bibinfo{person}{Yehuda Lindell}, \bibinfo{person}{Ariel Nof}, {and} \bibinfo{person}{Kazuma Ohara}.} \bibinfo{year}{2016}\natexlab{}.
\newblock \showarticletitle{High-throughput semi-honest secure three-party computation with an honest majority}. In \bibinfo{booktitle}{\emph{Proceedings of the 2016 ACM SIGSAC Conference on Computer and Communications Security}}. \bibinfo{pages}{805--817}.
\newblock


\bibitem[\protect\citeauthoryear{Asharov, Chan, Nayak, Pass, Ren, and Shi}{Asharov et~al\mbox{.}}{2020}]%
        {asharov2020bucket}
\bibfield{author}{\bibinfo{person}{Gilad Asharov}, \bibinfo{person}{TH~Hubert Chan}, \bibinfo{person}{Kartik Nayak}, \bibinfo{person}{Rafael Pass}, \bibinfo{person}{Ling Ren}, {and} \bibinfo{person}{Elaine Shi}.} \bibinfo{year}{2020}\natexlab{}.
\newblock \showarticletitle{Bucket oblivious sort: An extremely simple oblivious sort}. In \bibinfo{booktitle}{\emph{Symposium on Simplicity in Algorithms}}. SIAM, \bibinfo{pages}{8--14}.
\newblock


\bibitem[\protect\citeauthoryear{Batcher}{Batcher}{1968}]%
        {batcher1968sorting}
\bibfield{author}{\bibinfo{person}{Kenneth~E Batcher}.} \bibinfo{year}{1968}\natexlab{}.
\newblock \showarticletitle{Sorting networks and their applications}. In \bibinfo{booktitle}{\emph{Proceedings of the April 30--May 2, 1968, spring joint computer conference}}. \bibinfo{pages}{307--314}.
\newblock


\bibitem[\protect\citeauthoryear{Bater, Elliott, Eggen, Goel, Kho, and Rogers}{Bater et~al\mbox{.}}{2017}]%
        {bater2017smcql}
\bibfield{author}{\bibinfo{person}{Johes Bater}, \bibinfo{person}{Gregory Elliott}, \bibinfo{person}{Craig Eggen}, \bibinfo{person}{Satyender Goel}, \bibinfo{person}{Abel~N Kho}, {and} \bibinfo{person}{Jennie Rogers}.} \bibinfo{year}{2017}\natexlab{}.
\newblock \showarticletitle{SMCQL: Secure Query Processing for Private Data Networks.}
\newblock \bibinfo{journal}{\emph{Proc. VLDB Endow.}} \bibinfo{volume}{10}, \bibinfo{number}{6} (\bibinfo{year}{2017}), \bibinfo{pages}{673--684}.
\newblock


\bibitem[\protect\citeauthoryear{Bater, He, Ehrich, Machanavajjhala, and Rogers}{Bater et~al\mbox{.}}{2018}]%
        {bater2018shrinkwrap}
\bibfield{author}{\bibinfo{person}{Johes Bater}, \bibinfo{person}{Xi He}, \bibinfo{person}{William Ehrich}, \bibinfo{person}{Ashwin Machanavajjhala}, {and} \bibinfo{person}{Jennie Rogers}.} \bibinfo{year}{2018}\natexlab{}.
\newblock \showarticletitle{Shrinkwrap: efficient sql query processing in differentially private data federations}.
\newblock \bibinfo{journal}{\emph{Proceedings of the VLDB Endowment}} \bibinfo{volume}{12}, \bibinfo{number}{3} (\bibinfo{year}{2018}).
\newblock


\bibitem[\protect\citeauthoryear{Bater, Park, He, Wang, and Rogers}{Bater et~al\mbox{.}}{2020}]%
        {bater2020saqe}
\bibfield{author}{\bibinfo{person}{Johes Bater}, \bibinfo{person}{Yongjoo Park}, \bibinfo{person}{Xi He}, \bibinfo{person}{Xiao Wang}, {and} \bibinfo{person}{Jennie Rogers}.} \bibinfo{year}{2020}\natexlab{}.
\newblock \showarticletitle{Saqe: practical privacy-preserving approximate query processing for data federations}.
\newblock \bibinfo{journal}{\emph{Proceedings of the VLDB Endowment}} \bibinfo{volume}{13}, \bibinfo{number}{12} (\bibinfo{year}{2020}), \bibinfo{pages}{2691--2705}.
\newblock


\bibitem[\protect\citeauthoryear{Ben-Or, Goldwasser, and Wigderson}{Ben-Or et~al\mbox{.}}{2019}]%
        {ben2019completeness}
\bibfield{author}{\bibinfo{person}{Michael Ben-Or}, \bibinfo{person}{Shafi Goldwasser}, {and} \bibinfo{person}{Avi Wigderson}.} \bibinfo{year}{2019}\natexlab{}.
\newblock \showarticletitle{Completeness theorems for non-cryptographic fault-tolerant distributed computation}.
\newblock In \bibinfo{booktitle}{\emph{Providing sound foundations for cryptography: on the work of Shafi Goldwasser and Silvio Micali}}. \bibinfo{pages}{351--371}.
\newblock


\bibitem[\protect\citeauthoryear{Blackstone, Kamara, and Moataz}{Blackstone et~al\mbox{.}}{2019}]%
        {blackstone2019revisiting}
\bibfield{author}{\bibinfo{person}{Laura Blackstone}, \bibinfo{person}{Seny Kamara}, {and} \bibinfo{person}{Tarik Moataz}.} \bibinfo{year}{2019}\natexlab{}.
\newblock \showarticletitle{Revisiting leakage abuse attacks}.
\newblock \bibinfo{journal}{\emph{Cryptology ePrint Archive}} (\bibinfo{year}{2019}).
\newblock


\bibitem[\protect\citeauthoryear{Blasgen, Astrahan, Chamberlin, Gray, King, Lindsay, Lorie, Mehl, Price, Putzolu, et~al\mbox{.}}{Blasgen et~al\mbox{.}}{1981}]%
        {blasgen1981system}
\bibfield{author}{\bibinfo{person}{Mike~W Blasgen}, \bibinfo{person}{Morton~M Astrahan}, \bibinfo{person}{Donald~D Chamberlin}, \bibinfo{person}{JN Gray}, \bibinfo{person}{WF King}, \bibinfo{person}{Bruce~G Lindsay}, \bibinfo{person}{Raymond~A Lorie}, \bibinfo{person}{James~W Mehl}, \bibinfo{person}{Thomas~G Price}, \bibinfo{person}{Gianfranco~R Putzolu}, {et~al\mbox{.}}} \bibinfo{year}{1981}\natexlab{}.
\newblock \showarticletitle{System R: An architectural overview}.
\newblock \bibinfo{journal}{\emph{IBM systems journal}} \bibinfo{volume}{20}, \bibinfo{number}{1} (\bibinfo{year}{1981}), \bibinfo{pages}{41--62}.
\newblock


\bibitem[\protect\citeauthoryear{Bogatov, Kellaris, Kollios, Nissim, and O’Neill}{Bogatov et~al\mbox{.}}{2021}]%
        {bogatov2021epsolute}
\bibfield{author}{\bibinfo{person}{Dmytro Bogatov}, \bibinfo{person}{Georgios Kellaris}, \bibinfo{person}{George Kollios}, \bibinfo{person}{Kobbi Nissim}, {and} \bibinfo{person}{Adam O’Neill}.} \bibinfo{year}{2021}\natexlab{}.
\newblock \showarticletitle{Epsolute: E iciently erying Databases While Providing Differential Privacy}.
\newblock  (\bibinfo{year}{2021}).
\newblock


\bibitem[\protect\citeauthoryear{Bun and Steinke}{Bun and Steinke}{2016}]%
        {bun2016concentrated}
\bibfield{author}{\bibinfo{person}{Mark Bun} {and} \bibinfo{person}{Thomas Steinke}.} \bibinfo{year}{2016}\natexlab{}.
\newblock \showarticletitle{Concentrated differential privacy: Simplifications, extensions, and lower bounds}. In \bibinfo{booktitle}{\emph{Theory of Cryptography Conference}}. Springer, \bibinfo{pages}{635--658}.
\newblock


\bibitem[\protect\citeauthoryear{Cash, Grubbs, Perry, and Ristenpart}{Cash et~al\mbox{.}}{2015}]%
        {cash2015leakage}
\bibfield{author}{\bibinfo{person}{David Cash}, \bibinfo{person}{Paul Grubbs}, \bibinfo{person}{Jason Perry}, {and} \bibinfo{person}{Thomas Ristenpart}.} \bibinfo{year}{2015}\natexlab{}.
\newblock \showarticletitle{Leakage-abuse attacks against searchable encryption}. In \bibinfo{booktitle}{\emph{Proceedings of the 22nd ACM SIGSAC conference on computer and communications security}}. \bibinfo{pages}{668--679}.
\newblock


\bibitem[\protect\citeauthoryear{Chan, Chung, Maggs, and Shi}{Chan et~al\mbox{.}}{2022}]%
        {chan2022foundations}
\bibfield{author}{\bibinfo{person}{T-H~Hubert Chan}, \bibinfo{person}{Kai-Min Chung}, \bibinfo{person}{Bruce Maggs}, {and} \bibinfo{person}{Elaine Shi}.} \bibinfo{year}{2022}\natexlab{}.
\newblock \showarticletitle{Foundations of differentially oblivious algorithms}.
\newblock \bibinfo{journal}{\emph{ACM Journal of the ACM (JACM)}} \bibinfo{volume}{69}, \bibinfo{number}{4} (\bibinfo{year}{2022}), \bibinfo{pages}{1--49}.
\newblock


\bibitem[\protect\citeauthoryear{Chang, Xie, Wang, and Li}{Chang et~al\mbox{.}}{2022}]%
        {chang2022towards}
\bibfield{author}{\bibinfo{person}{Zhao Chang}, \bibinfo{person}{Dong Xie}, \bibinfo{person}{Sheng Wang}, {and} \bibinfo{person}{Feifei Li}.} \bibinfo{year}{2022}\natexlab{}.
\newblock \showarticletitle{Towards Practical Oblivious Join}. In \bibinfo{booktitle}{\emph{Proceedings of the 2022 International Conference on Management of Data}}. \bibinfo{pages}{803--817}.
\newblock


\bibitem[\protect\citeauthoryear{Chenghong~Wang}{Chenghong~Wang}{2024}]%
        {special2024online}
\bibfield{author}{\bibinfo{person}{Johes Bater Yukui~Luo Chenghong~Wang, Lina~Qiu}.} \bibinfo{year}{2024}\natexlab{}.
\newblock \bibinfo{title}{SPECIAL: Synopsis Assisted Secure Collaborative Analytics}.
\newblock \bibinfo{howpublished}{\url{https://arxiv.org/abs/2404.18388}}.
\newblock


\bibitem[\protect\citeauthoryear{Chu, Zhuo, Shi, and Chan}{Chu et~al\mbox{.}}{2021}]%
        {chu2021differentially}
\bibfield{author}{\bibinfo{person}{Shumo Chu}, \bibinfo{person}{Danyang Zhuo}, \bibinfo{person}{Elaine Shi}, {and} \bibinfo{person}{TH~Hubert Chan}.} \bibinfo{year}{2021}\natexlab{}.
\newblock \showarticletitle{Differentially oblivious database joins: Overcoming the worst-case curse of fully oblivious algorithms}.
\newblock \bibinfo{journal}{\emph{Cryptology ePrint Archive}} (\bibinfo{year}{2021}).
\newblock


\bibitem[\protect\citeauthoryear{Crooks, Burke, Cecchetti, Harel, Agarwal, and Alvisi}{Crooks et~al\mbox{.}}{2018}]%
        {crooks2018obladi}
\bibfield{author}{\bibinfo{person}{Natacha Crooks}, \bibinfo{person}{Matthew Burke}, \bibinfo{person}{Ethan Cecchetti}, \bibinfo{person}{Sitar Harel}, \bibinfo{person}{Rachit Agarwal}, {and} \bibinfo{person}{Lorenzo Alvisi}.} \bibinfo{year}{2018}\natexlab{}.
\newblock \showarticletitle{Obladi: Oblivious serializable transactions in the cloud}. In \bibinfo{booktitle}{\emph{13th USENIX Symposium on Operating Systems Design and Implementation (OSDI 18)}}. \bibinfo{pages}{727--743}.
\newblock


\bibitem[\protect\citeauthoryear{Dong, Fang, Yi, Tao, and Machanavajjhala}{Dong et~al\mbox{.}}{2022}]%
        {dong2022r2t}
\bibfield{author}{\bibinfo{person}{Wei Dong}, \bibinfo{person}{Juanru Fang}, \bibinfo{person}{Ke Yi}, \bibinfo{person}{Yuchao Tao}, {and} \bibinfo{person}{Ashwin Machanavajjhala}.} \bibinfo{year}{2022}\natexlab{}.
\newblock \showarticletitle{R2t: Instance-optimal truncation for differentially private query evaluation with foreign keys}. In \bibinfo{booktitle}{\emph{Proceedings of the 2022 International Conference on Management of Data}}. \bibinfo{pages}{759--772}.
\newblock


\bibitem[\protect\citeauthoryear{Dwork, Naor, Pitassi, and Rothblum}{Dwork et~al\mbox{.}}{2010}]%
        {dwork2010differential}
\bibfield{author}{\bibinfo{person}{Cynthia Dwork}, \bibinfo{person}{Moni Naor}, \bibinfo{person}{Toniann Pitassi}, {and} \bibinfo{person}{Guy~N Rothblum}.} \bibinfo{year}{2010}\natexlab{}.
\newblock \showarticletitle{Differential privacy under continual observation}. In \bibinfo{booktitle}{\emph{Proceedings of the forty-second ACM symposium on Theory of computing}}. \bibinfo{pages}{715--724}.
\newblock


\bibitem[\protect\citeauthoryear{Dwork, Roth, et~al\mbox{.}}{Dwork et~al\mbox{.}}{2014}]%
        {dwork2014algorithmic}
\bibfield{author}{\bibinfo{person}{Cynthia Dwork}, \bibinfo{person}{Aaron Roth}, {et~al\mbox{.}}} \bibinfo{year}{2014}\natexlab{}.
\newblock \showarticletitle{The algorithmic foundations of differential privacy.}
\newblock \bibinfo{journal}{\emph{Found. Trends Theor. Comput. Sci.}} \bibinfo{volume}{9}, \bibinfo{number}{3-4} (\bibinfo{year}{2014}), \bibinfo{pages}{211--407}.
\newblock


\bibitem[\protect\citeauthoryear{Elmasri, Navathe, Elmasri, and Navathe}{Elmasri et~al\mbox{.}}{2020}]%
        {elmasri2020fundamentals}
\bibfield{author}{\bibinfo{person}{R Elmasri}, \bibinfo{person}{SB Navathe}, \bibinfo{person}{R Elmasri}, {and} \bibinfo{person}{SB Navathe}.} \bibinfo{year}{2020}\natexlab{}.
\newblock \showarticletitle{Fundamentals of Database Systems</Title}. In \bibinfo{booktitle}{\emph{Advances in Databases and Information Systems: 24th European Conference, ADBIS 2020, Lyon, France, August 25--27, 2020, Proceedings}}, Vol.~\bibinfo{volume}{12245}. Springer Nature, \bibinfo{pages}{139}.
\newblock


\bibitem[\protect\citeauthoryear{Eskandarian and Zaharia}{Eskandarian and Zaharia}{2017}]%
        {eskandarian2017oblidb}
\bibfield{author}{\bibinfo{person}{Saba Eskandarian} {and} \bibinfo{person}{Matei Zaharia}.} \bibinfo{year}{2017}\natexlab{}.
\newblock \showarticletitle{Oblidb: Oblivious query processing for secure databases}.
\newblock \bibinfo{journal}{\emph{arXiv preprint arXiv:1710.00458}} (\bibinfo{year}{2017}).
\newblock


\bibitem[\protect\citeauthoryear{Goldreich}{Goldreich}{2009}]%
        {goldreich2009foundations}
\bibfield{author}{\bibinfo{person}{Oded Goldreich}.} \bibinfo{year}{2009}\natexlab{}.
\newblock \bibinfo{booktitle}{\emph{Foundations of cryptography: volume 2, basic applications}}.
\newblock \bibinfo{publisher}{Cambridge university press}.
\newblock


\bibitem[\protect\citeauthoryear{Goodrich}{Goodrich}{2014}]%
        {goodrich2014zig}
\bibfield{author}{\bibinfo{person}{Michael~T Goodrich}.} \bibinfo{year}{2014}\natexlab{}.
\newblock \showarticletitle{Zig-zag sort: A simple deterministic data-oblivious sorting algorithm running in o (n log n) time}. In \bibinfo{booktitle}{\emph{Proceedings of the forty-sixth annual ACM symposium on Theory of computing}}. \bibinfo{pages}{684--693}.
\newblock


\bibitem[\protect\citeauthoryear{Groce, Rindal, and Rosulek}{Groce et~al\mbox{.}}{2019}]%
        {groce2019cheaper}
\bibfield{author}{\bibinfo{person}{Adam Groce}, \bibinfo{person}{Peter Rindal}, {and} \bibinfo{person}{Mike Rosulek}.} \bibinfo{year}{2019}\natexlab{}.
\newblock \showarticletitle{Cheaper private set intersection via differentially private leakage}.
\newblock \bibinfo{journal}{\emph{Proceedings on Privacy Enhancing Technologies}} \bibinfo{volume}{2019}, \bibinfo{number}{3} (\bibinfo{year}{2019}).
\newblock


\bibitem[\protect\citeauthoryear{Grubbs, Lacharit{\'e}, Minaud, and Paterson}{Grubbs et~al\mbox{.}}{2018}]%
        {grubbs2018pump}
\bibfield{author}{\bibinfo{person}{Paul Grubbs}, \bibinfo{person}{Marie-Sarah Lacharit{\'e}}, \bibinfo{person}{Brice Minaud}, {and} \bibinfo{person}{Kenneth~G Paterson}.} \bibinfo{year}{2018}\natexlab{}.
\newblock \showarticletitle{Pump up the volume: Practical database reconstruction from volume leakage on range queries}. In \bibinfo{booktitle}{\emph{Proceedings of the 2018 ACM SIGSAC Conference on Computer and Communications Security}}. \bibinfo{pages}{315--331}.
\newblock


\bibitem[\protect\citeauthoryear{Han, Wu, Wu, Zhu, Yang, Tan, Zeng, Cong, Qin, Pfadler, et~al\mbox{.}}{Han et~al\mbox{.}}{2021}]%
        {han2021cardinality}
\bibfield{author}{\bibinfo{person}{Yuxing Han}, \bibinfo{person}{Ziniu Wu}, \bibinfo{person}{Peizhi Wu}, \bibinfo{person}{Rong Zhu}, \bibinfo{person}{Jingyi Yang}, \bibinfo{person}{Liang~Wei Tan}, \bibinfo{person}{Kai Zeng}, \bibinfo{person}{Gao Cong}, \bibinfo{person}{Yanzhao Qin}, \bibinfo{person}{Andreas Pfadler}, {et~al\mbox{.}}} \bibinfo{year}{2021}\natexlab{}.
\newblock \showarticletitle{Cardinality estimation in DBMS: A comprehensive benchmark evaluation}.
\newblock \bibinfo{journal}{\emph{arXiv preprint arXiv:2109.05877}} (\bibinfo{year}{2021}).
\newblock


\bibitem[\protect\citeauthoryear{Harmouch and Naumann}{Harmouch and Naumann}{2017}]%
        {harmouch2017cardinality}
\bibfield{author}{\bibinfo{person}{Hazar Harmouch} {and} \bibinfo{person}{Felix Naumann}.} \bibinfo{year}{2017}\natexlab{}.
\newblock \showarticletitle{Cardinality estimation: An experimental survey}.
\newblock \bibinfo{journal}{\emph{Proceedings of the VLDB Endowment}} \bibinfo{volume}{11}, \bibinfo{number}{4} (\bibinfo{year}{2017}), \bibinfo{pages}{499--512}.
\newblock


\bibitem[\protect\citeauthoryear{He, Machanavajjhala, Flynn, and Srivastava}{He et~al\mbox{.}}{2017}]%
        {he2017composing}
\bibfield{author}{\bibinfo{person}{Xi He}, \bibinfo{person}{Ashwin Machanavajjhala}, \bibinfo{person}{Cheryl Flynn}, {and} \bibinfo{person}{Divesh Srivastava}.} \bibinfo{year}{2017}\natexlab{}.
\newblock \showarticletitle{Composing differential privacy and secure computation: A case study on scaling private record linkage}. In \bibinfo{booktitle}{\emph{Proceedings of the 2017 ACM SIGSAC Conference on Computer and Communications Security}}. \bibinfo{pages}{1389--1406}.
\newblock


\bibitem[\protect\citeauthoryear{He, Wong, Kao, Cheung, Li, Yiu, and Lo}{He et~al\mbox{.}}{2015}]%
        {he2015sdb}
\bibfield{author}{\bibinfo{person}{Zhian He}, \bibinfo{person}{Wai~Kit Wong}, \bibinfo{person}{Ben Kao}, \bibinfo{person}{David Wai~Lok Cheung}, \bibinfo{person}{Rongbin Li}, \bibinfo{person}{Siu~Ming Yiu}, {and} \bibinfo{person}{Eric Lo}.} \bibinfo{year}{2015}\natexlab{}.
\newblock \showarticletitle{Sdb: A secure query processing system with data interoperability}.
\newblock \bibinfo{journal}{\emph{Proceedings of the VLDB Endowment}} \bibinfo{volume}{8}, \bibinfo{number}{12} (\bibinfo{year}{2015}), \bibinfo{pages}{1876--1879}.
\newblock


\bibitem[\protect\citeauthoryear{Hertzschuch, Hartmann, Habich, and Lehner}{Hertzschuch et~al\mbox{.}}{2021}]%
        {hertzschuch2021simplicity}
\bibfield{author}{\bibinfo{person}{Axel Hertzschuch}, \bibinfo{person}{Claudio Hartmann}, \bibinfo{person}{Dirk Habich}, {and} \bibinfo{person}{Wolfgang Lehner}.} \bibinfo{year}{2021}\natexlab{}.
\newblock \showarticletitle{Simplicity Done Right for Join Ordering.}. In \bibinfo{booktitle}{\emph{CIDR}}.
\newblock


\bibitem[\protect\citeauthoryear{Huang, Liu, Cui, Fang, Ma, Xu, Shen, Tang, Zhou, Huang, et~al\mbox{.}}{Huang et~al\mbox{.}}{2020}]%
        {huang2020tidb}
\bibfield{author}{\bibinfo{person}{Dongxu Huang}, \bibinfo{person}{Qi Liu}, \bibinfo{person}{Qiu Cui}, \bibinfo{person}{Zhuhe Fang}, \bibinfo{person}{Xiaoyu Ma}, \bibinfo{person}{Fei Xu}, \bibinfo{person}{Li Shen}, \bibinfo{person}{Liu Tang}, \bibinfo{person}{Yuxing Zhou}, \bibinfo{person}{Menglong Huang}, {et~al\mbox{.}}} \bibinfo{year}{2020}\natexlab{}.
\newblock \showarticletitle{TiDB: a Raft-based HTAP database}.
\newblock \bibinfo{journal}{\emph{Proceedings of the VLDB Endowment}} \bibinfo{volume}{13}, \bibinfo{number}{12} (\bibinfo{year}{2020}), \bibinfo{pages}{3072--3084}.
\newblock


\bibitem[\protect\citeauthoryear{J{\"o}nsson, Kreitz, and Uddin}{J{\"o}nsson et~al\mbox{.}}{2011}]%
        {jonsson2011secure}
\bibfield{author}{\bibinfo{person}{Kristj{\"a}n~Valur J{\"o}nsson}, \bibinfo{person}{Gunnar Kreitz}, {and} \bibinfo{person}{Misbah Uddin}.} \bibinfo{year}{2011}\natexlab{}.
\newblock \showarticletitle{Secure multi-party sorting and applications}.
\newblock \bibinfo{journal}{\emph{Cryptology ePrint Archive}} (\bibinfo{year}{2011}).
\newblock


\bibitem[\protect\citeauthoryear{Kamara, Mohassel, and Raykova}{Kamara et~al\mbox{.}}{2011}]%
        {kamara2011outsourcing}
\bibfield{author}{\bibinfo{person}{Seny Kamara}, \bibinfo{person}{Payman Mohassel}, {and} \bibinfo{person}{Mariana Raykova}.} \bibinfo{year}{2011}\natexlab{}.
\newblock \showarticletitle{Outsourcing multi-party computation}.
\newblock \bibinfo{journal}{\emph{Cryptology ePrint Archive}} (\bibinfo{year}{2011}).
\newblock


\bibitem[\protect\citeauthoryear{Kellaris, Kollios, Nissim, and O'neill}{Kellaris et~al\mbox{.}}{2016}]%
        {kellaris2016generic}
\bibfield{author}{\bibinfo{person}{Georgios Kellaris}, \bibinfo{person}{George Kollios}, \bibinfo{person}{Kobbi Nissim}, {and} \bibinfo{person}{Adam O'neill}.} \bibinfo{year}{2016}\natexlab{}.
\newblock \showarticletitle{Generic attacks on secure outsourced databases}. In \bibinfo{booktitle}{\emph{Proceedings of the 2016 ACM SIGSAC Conference on Computer and Communications Security}}. \bibinfo{pages}{1329--1340}.
\newblock


\bibitem[\protect\citeauthoryear{Keller and Scholl}{Keller and Scholl}{2014}]%
        {keller2014efficient}
\bibfield{author}{\bibinfo{person}{Marcel Keller} {and} \bibinfo{person}{Peter Scholl}.} \bibinfo{year}{2014}\natexlab{}.
\newblock \showarticletitle{Efficient, oblivious data structures for MPC}. In \bibinfo{booktitle}{\emph{Advances in Cryptology--ASIACRYPT 2014: 20th International Conference on the Theory and Application of Cryptology and Information Security, Kaoshiung, Taiwan, ROC, December 7-11, 2014, Proceedings, Part II 20}}. Springer, \bibinfo{pages}{506--525}.
\newblock


\bibitem[\protect\citeauthoryear{Kifer and Machanavajjhala}{Kifer and Machanavajjhala}{2011}]%
        {kifer2011no}
\bibfield{author}{\bibinfo{person}{Daniel Kifer} {and} \bibinfo{person}{Ashwin Machanavajjhala}.} \bibinfo{year}{2011}\natexlab{}.
\newblock \showarticletitle{No free lunch in data privacy}. In \bibinfo{booktitle}{\emph{Proceedings of the 2011 ACM SIGMOD International Conference on Management of data}}. \bibinfo{pages}{193--204}.
\newblock


\bibitem[\protect\citeauthoryear{Knott, Venkataraman, Hannun, Sengupta, Ibrahim, and van~der Maaten}{Knott et~al\mbox{.}}{2021}]%
        {knott2021crypten}
\bibfield{author}{\bibinfo{person}{Brian Knott}, \bibinfo{person}{Shobha Venkataraman}, \bibinfo{person}{Awni Hannun}, \bibinfo{person}{Shubho Sengupta}, \bibinfo{person}{Mark Ibrahim}, {and} \bibinfo{person}{Laurens van~der Maaten}.} \bibinfo{year}{2021}\natexlab{}.
\newblock \showarticletitle{Crypten: Secure multi-party computation meets machine learning}.
\newblock \bibinfo{journal}{\emph{Advances in Neural Information Processing Systems}}  \bibinfo{volume}{34} (\bibinfo{year}{2021}), \bibinfo{pages}{4961--4973}.
\newblock


\bibitem[\protect\citeauthoryear{Kotsogiannis, Tao, He, Fanaeepour, Machanavajjhala, Hay, and Miklau}{Kotsogiannis et~al\mbox{.}}{2019}]%
        {kotsogiannis2019privatesql}
\bibfield{author}{\bibinfo{person}{Ios Kotsogiannis}, \bibinfo{person}{Yuchao Tao}, \bibinfo{person}{Xi He}, \bibinfo{person}{Maryam Fanaeepour}, \bibinfo{person}{Ashwin Machanavajjhala}, \bibinfo{person}{Michael Hay}, {and} \bibinfo{person}{Gerome Miklau}.} \bibinfo{year}{2019}\natexlab{}.
\newblock \showarticletitle{Privatesql: a differentially private sql query engine}.
\newblock \bibinfo{journal}{\emph{Proceedings of the VLDB Endowment}} \bibinfo{volume}{12}, \bibinfo{number}{11} (\bibinfo{year}{2019}), \bibinfo{pages}{1371--1384}.
\newblock


\bibitem[\protect\citeauthoryear{Kraska, Beutel, Chi, Dean, and Polyzotis}{Kraska et~al\mbox{.}}{2018}]%
        {kraska2018case}
\bibfield{author}{\bibinfo{person}{Tim Kraska}, \bibinfo{person}{Alex Beutel}, \bibinfo{person}{Ed~H Chi}, \bibinfo{person}{Jeffrey Dean}, {and} \bibinfo{person}{Neoklis Polyzotis}.} \bibinfo{year}{2018}\natexlab{}.
\newblock \showarticletitle{The case for learned index structures}. In \bibinfo{booktitle}{\emph{Proceedings of the 2018 international conference on management of data}}. \bibinfo{pages}{489--504}.
\newblock


\bibitem[\protect\citeauthoryear{Liagouris, Kalavri, Faisal, and Varia}{Liagouris et~al\mbox{.}}{2023}]%
        {liagouris2023secrecy}
\bibfield{author}{\bibinfo{person}{John Liagouris}, \bibinfo{person}{Vasiliki Kalavri}, \bibinfo{person}{Muhammad Faisal}, {and} \bibinfo{person}{Mayank Varia}.} \bibinfo{year}{2023}\natexlab{}.
\newblock \showarticletitle{$\{$SECRECY$\}$: Secure collaborative analytics in untrusted clouds}. In \bibinfo{booktitle}{\emph{20th USENIX Symposium on Networked Systems Design and Implementation (NSDI 23)}}. \bibinfo{pages}{1031--1056}.
\newblock


\bibitem[\protect\citeauthoryear{Lovingmage}{Lovingmage}{2024}]%
        {special2024}
\bibfield{author}{\bibinfo{person}{Lovingmage}.} \bibinfo{year}{2024}\natexlab{}.
\newblock \bibinfo{title}{Synopsis Assisted Secure Collaborative Analytics}.
\newblock \bibinfo{howpublished}{\url{https://github.com/lovingmage/SPECIAL/}}.
\newblock


\bibitem[\protect\citeauthoryear{Micali, Goldreich, and Wigderson}{Micali et~al\mbox{.}}{1987}]%
        {micali1987play}
\bibfield{author}{\bibinfo{person}{Silvio Micali}, \bibinfo{person}{Oded Goldreich}, {and} \bibinfo{person}{Avi Wigderson}.} \bibinfo{year}{1987}\natexlab{}.
\newblock \showarticletitle{How to play any mental game}. In \bibinfo{booktitle}{\emph{Proceedings of the Nineteenth ACM Symp. on Theory of Computing, STOC}}. ACM New York, NY, USA, \bibinfo{pages}{218--229}.
\newblock


\bibitem[\protect\citeauthoryear{Mishra, Poddar, Chen, Chiesa, and Popa}{Mishra et~al\mbox{.}}{2018}]%
        {mishra2018oblix}
\bibfield{author}{\bibinfo{person}{Pratyush Mishra}, \bibinfo{person}{Rishabh Poddar}, \bibinfo{person}{Jerry Chen}, \bibinfo{person}{Alessandro Chiesa}, {and} \bibinfo{person}{Raluca~Ada Popa}.} \bibinfo{year}{2018}\natexlab{}.
\newblock \showarticletitle{Oblix: An efficient oblivious search index}. In \bibinfo{booktitle}{\emph{2018 IEEE symposium on security and privacy (SP)}}. IEEE, \bibinfo{pages}{279--296}.
\newblock


\bibitem[\protect\citeauthoryear{Mohassel and Zhang}{Mohassel and Zhang}{2017}]%
        {mohassel2017secureml}
\bibfield{author}{\bibinfo{person}{Payman Mohassel} {and} \bibinfo{person}{Yupeng Zhang}.} \bibinfo{year}{2017}\natexlab{}.
\newblock \showarticletitle{Secureml: A system for scalable privacy-preserving machine learning}. In \bibinfo{booktitle}{\emph{2017 IEEE symposium on security and privacy (SP)}}. IEEE, \bibinfo{pages}{19--38}.
\newblock


\bibitem[\protect\citeauthoryear{Mouris, Masny, Trieu, Sengupta, Buddhavarapu, and Case}{Mouris et~al\mbox{.}}{2024}]%
        {mouris2024delegated}
\bibfield{author}{\bibinfo{person}{Dimitris Mouris}, \bibinfo{person}{Daniel Masny}, \bibinfo{person}{Ni Trieu}, \bibinfo{person}{Shubho Sengupta}, \bibinfo{person}{Prasad Buddhavarapu}, {and} \bibinfo{person}{Benjamin Case}.} \bibinfo{year}{2024}\natexlab{}.
\newblock \showarticletitle{Delegated Private Matching for Compute}.
\newblock \bibinfo{journal}{\emph{Proceedings on Privacy Enhancing Technologies}} (\bibinfo{year}{2024}).
\newblock


\bibitem[\protect\citeauthoryear{Oya and Kerschbaum}{Oya and Kerschbaum}{2021}]%
        {oya2021hiding}
\bibfield{author}{\bibinfo{person}{Simon Oya} {and} \bibinfo{person}{Florian Kerschbaum}.} \bibinfo{year}{2021}\natexlab{}.
\newblock \showarticletitle{Hiding the access pattern is not enough: Exploiting search pattern leakage in searchable encryption}. In \bibinfo{booktitle}{\emph{30th USENIX security symposium (USENIX Security 21)}}. \bibinfo{pages}{127--142}.
\newblock


\bibitem[\protect\citeauthoryear{Patel, Persiano, Yeo, and Yung}{Patel et~al\mbox{.}}{2019}]%
        {patel2019mitigating}
\bibfield{author}{\bibinfo{person}{Sarvar Patel}, \bibinfo{person}{Giuseppe Persiano}, \bibinfo{person}{Kevin Yeo}, {and} \bibinfo{person}{Moti Yung}.} \bibinfo{year}{2019}\natexlab{}.
\newblock \showarticletitle{Mitigating leakage in secure cloud-hosted data structures: Volume-hiding for multi-maps via hashing}. In \bibinfo{booktitle}{\emph{Proceedings of the 2019 ACM SIGSAC conference on computer and communications security}}. \bibinfo{pages}{79--93}.
\newblock


\bibitem[\protect\citeauthoryear{Pettai and Laud}{Pettai and Laud}{2015}]%
        {pettai2015combining}
\bibfield{author}{\bibinfo{person}{Martin Pettai} {and} \bibinfo{person}{Peeter Laud}.} \bibinfo{year}{2015}\natexlab{}.
\newblock \showarticletitle{Combining differential privacy and secure multiparty computation}. In \bibinfo{booktitle}{\emph{Proceedings of the 31st annual computer security applications conference}}. \bibinfo{pages}{421--430}.
\newblock


\bibitem[\protect\citeauthoryear{Pinkas and Reinman}{Pinkas and Reinman}{2010}]%
        {pinkas2010oblivious}
\bibfield{author}{\bibinfo{person}{Benny Pinkas} {and} \bibinfo{person}{Tzachy Reinman}.} \bibinfo{year}{2010}\natexlab{}.
\newblock \showarticletitle{Oblivious RAM revisited}. In \bibinfo{booktitle}{\emph{Advances in Cryptology--CRYPTO 2010: 30th Annual Cryptology Conference, Santa Barbara, CA, USA, August 15-19, 2010. Proceedings 30}}. Springer, \bibinfo{pages}{502--519}.
\newblock


\bibitem[\protect\citeauthoryear{Poddar, Kalra, Yanai, Deng, Popa, and Hellerstein}{Poddar et~al\mbox{.}}{2021}]%
        {poddar2021senate}
\bibfield{author}{\bibinfo{person}{Rishabh Poddar}, \bibinfo{person}{Sukrit Kalra}, \bibinfo{person}{Avishay Yanai}, \bibinfo{person}{Ryan Deng}, \bibinfo{person}{Raluca~Ada Popa}, {and} \bibinfo{person}{Joseph~M Hellerstein}.} \bibinfo{year}{2021}\natexlab{}.
\newblock \showarticletitle{Senate: a $\{$Maliciously-Secure$\}$$\{$MPC$\}$ platform for collaborative analytics}. In \bibinfo{booktitle}{\emph{30th USENIX Security Symposium (USENIX Security 21)}}. \bibinfo{pages}{2129--2146}.
\newblock


\bibitem[\protect\citeauthoryear{Qin, Jayaram, Shi, Song, Zhuo, and Chu}{Qin et~al\mbox{.}}{2022}]%
        {qin2022adore}
\bibfield{author}{\bibinfo{person}{Lianke Qin}, \bibinfo{person}{Rajesh Jayaram}, \bibinfo{person}{Elaine Shi}, \bibinfo{person}{Zhao Song}, \bibinfo{person}{Danyang Zhuo}, {and} \bibinfo{person}{Shumo Chu}.} \bibinfo{year}{2022}\natexlab{}.
\newblock \showarticletitle{Adore: Differentially oblivious relational database operators}.
\newblock \bibinfo{journal}{\emph{arXiv preprint arXiv:2212.05176}} (\bibinfo{year}{2022}).
\newblock


\bibitem[\protect\citeauthoryear{Qiu, Kellaris, Mamoulis, Nissim, and Kollios}{Qiu et~al\mbox{.}}{2023}]%
        {qiudoquet}
\bibfield{author}{\bibinfo{person}{Lina Qiu}, \bibinfo{person}{Georgios Kellaris}, \bibinfo{person}{Nikos Mamoulis}, \bibinfo{person}{Kobbi Nissim}, {and} \bibinfo{person}{George Kollios}.} \bibinfo{year}{2023}\natexlab{}.
\newblock \showarticletitle{Doquet: Differentially Oblivious Range and Join Queries with Private Data Structures}.
\newblock \bibinfo{journal}{\emph{Proceedings of the VLDB Endowment}} \bibinfo{volume}{16}, \bibinfo{number}{13} (\bibinfo{year}{2023}), \bibinfo{pages}{4160--4173}.
\newblock


\bibitem[\protect\citeauthoryear{Roy~Chowdhury, Wang, He, Machanavajjhala, and Jha}{Roy~Chowdhury et~al\mbox{.}}{2020}]%
        {roy2020crypt}
\bibfield{author}{\bibinfo{person}{Amrita Roy~Chowdhury}, \bibinfo{person}{Chenghong Wang}, \bibinfo{person}{Xi He}, \bibinfo{person}{Ashwin Machanavajjhala}, {and} \bibinfo{person}{Somesh Jha}.} \bibinfo{year}{2020}\natexlab{}.
\newblock \showarticletitle{Crypte: Crypto-assisted differential privacy on untrusted servers}. In \bibinfo{booktitle}{\emph{Proceedings of the 2020 ACM SIGMOD International Conference on Management of Data}}. \bibinfo{pages}{603--619}.
\newblock


\bibitem[\protect\citeauthoryear{Sasy, Johnson, and Goldberg}{Sasy et~al\mbox{.}}{2022}]%
        {sasy2022fast}
\bibfield{author}{\bibinfo{person}{Sajin Sasy}, \bibinfo{person}{Aaron Johnson}, {and} \bibinfo{person}{Ian Goldberg}.} \bibinfo{year}{2022}\natexlab{}.
\newblock \showarticletitle{Fast Fully Oblivious Compaction and Shuffling}. In \bibinfo{booktitle}{\emph{Proceedings of the 2022 ACM SIGSAC Conference on Computer and Communications Security}}. \bibinfo{pages}{2565--2579}.
\newblock


\bibitem[\protect\citeauthoryear{Scholl, Smart, and Wood}{Scholl et~al\mbox{.}}{2017}]%
        {scholl2017s}
\bibfield{author}{\bibinfo{person}{Peter Scholl}, \bibinfo{person}{Nigel~P Smart}, {and} \bibinfo{person}{Tim Wood}.} \bibinfo{year}{2017}\natexlab{}.
\newblock \showarticletitle{When it’s all just too much: outsourcing MPC-preprocessing}. In \bibinfo{booktitle}{\emph{Cryptography and Coding: 16th IMA International Conference, IMACC 2017, Oxford, UK, December 12-14, 2017, Proceedings 16}}. Springer, \bibinfo{pages}{77--99}.
\newblock


\bibitem[\protect\citeauthoryear{Shang, Oya, Peter, and Kerschbaum}{Shang et~al\mbox{.}}{2021}]%
        {shang2021obfuscated}
\bibfield{author}{\bibinfo{person}{Zhiwei Shang}, \bibinfo{person}{Simon Oya}, \bibinfo{person}{Andreas Peter}, {and} \bibinfo{person}{Florian Kerschbaum}.} \bibinfo{year}{2021}\natexlab{}.
\newblock \showarticletitle{Obfuscated access and search patterns in searchable encryption}.
\newblock \bibinfo{journal}{\emph{arXiv preprint arXiv:2102.09651}} (\bibinfo{year}{2021}).
\newblock


\bibitem[\protect\citeauthoryear{Silberschatz, Korth, and Sudarshan}{Silberschatz et~al\mbox{.}}{2011}]%
        {silberschatz2011database}
\bibfield{author}{\bibinfo{person}{Abraham Silberschatz}, \bibinfo{person}{Henry~F Korth}, {and} \bibinfo{person}{Shashank Sudarshan}.} \bibinfo{year}{2011}\natexlab{}.
\newblock \showarticletitle{Database system concepts}.
\newblock  (\bibinfo{year}{2011}).
\newblock


\bibitem[\protect\citeauthoryear{SMCQL}{SMCQL}{[n.d.]}]%
        {smcql-git}
\bibfield{author}{\bibinfo{person}{SMCQL}.} \bibinfo{year}{[n.d.]}\natexlab{}.
\newblock \bibinfo{title}{SMCQL: Secure Multi-party Computation Query Language}.
\newblock \bibinfo{howpublished}{\url{https://github.com/smcql/smcql/tree/master/conf/workload}}.
\newblock
\newblock
\shownote{Accessed: 2024-08-12.}


\bibitem[\protect\citeauthoryear{Stefanov, Shi, and Song}{Stefanov et~al\mbox{.}}{2011}]%
        {stefanov2011towards}
\bibfield{author}{\bibinfo{person}{Emil Stefanov}, \bibinfo{person}{Elaine Shi}, {and} \bibinfo{person}{Dawn Song}.} \bibinfo{year}{2011}\natexlab{}.
\newblock \showarticletitle{Towards practical oblivious RAM}.
\newblock \bibinfo{journal}{\emph{arXiv preprint arXiv:1106.3652}} (\bibinfo{year}{2011}).
\newblock


\bibitem[\protect\citeauthoryear{Tan, Knott, Tian, and Wu}{Tan et~al\mbox{.}}{2021}]%
        {tan2021cryptgpu}
\bibfield{author}{\bibinfo{person}{Sijun Tan}, \bibinfo{person}{Brian Knott}, \bibinfo{person}{Yuan Tian}, {and} \bibinfo{person}{David~J Wu}.} \bibinfo{year}{2021}\natexlab{}.
\newblock \showarticletitle{CryptGPU: Fast privacy-preserving machine learning on the GPU}. In \bibinfo{booktitle}{\emph{2021 IEEE Symposium on Security and Privacy (SP)}}. IEEE, \bibinfo{pages}{1021--1038}.
\newblock


\bibitem[\protect\citeauthoryear{Tinoco, Gao, and Shi}{Tinoco et~al\mbox{.}}{2023}]%
        {tinoco2023enigmap}
\bibfield{author}{\bibinfo{person}{Afonso Tinoco}, \bibinfo{person}{Sixiang Gao}, {and} \bibinfo{person}{Elaine Shi}.} \bibinfo{year}{2023}\natexlab{}.
\newblock \showarticletitle{$\{$EnigMap$\}$:$\{$External-Memory$\}$ Oblivious Map for Secure Enclaves}. In \bibinfo{booktitle}{\emph{32nd USENIX Security Symposium (USENIX Security 23)}}. \bibinfo{pages}{4033--4050}.
\newblock


\bibitem[\protect\citeauthoryear{Vadhan}{Vadhan}{2017}]%
        {vadhan2017complexity}
\bibfield{author}{\bibinfo{person}{Salil Vadhan}.} \bibinfo{year}{2017}\natexlab{}.
\newblock \showarticletitle{The complexity of differential privacy}.
\newblock \bibinfo{journal}{\emph{Tutorials on the Foundations of Cryptography: Dedicated to Oded Goldreich}} (\bibinfo{year}{2017}), \bibinfo{pages}{347--450}.
\newblock


\bibitem[\protect\citeauthoryear{Volgushev, Schwarzkopf, Getchell, Varia, Lapets, and Bestavros}{Volgushev et~al\mbox{.}}{2019}]%
        {volgushev2019conclave}
\bibfield{author}{\bibinfo{person}{Nikolaj Volgushev}, \bibinfo{person}{Malte Schwarzkopf}, \bibinfo{person}{Ben Getchell}, \bibinfo{person}{Mayank Varia}, \bibinfo{person}{Andrei Lapets}, {and} \bibinfo{person}{Azer Bestavros}.} \bibinfo{year}{2019}\natexlab{}.
\newblock \showarticletitle{Conclave: secure multi-party computation on big data}. In \bibinfo{booktitle}{\emph{Proceedings of the Fourteenth EuroSys Conference 2019}}. \bibinfo{pages}{1--18}.
\newblock


\bibitem[\protect\citeauthoryear{Wagh, Cuff, and Mittal}{Wagh et~al\mbox{.}}{2016}]%
        {wagh2016differentially}
\bibfield{author}{\bibinfo{person}{Sameer Wagh}, \bibinfo{person}{Paul Cuff}, {and} \bibinfo{person}{Prateek Mittal}.} \bibinfo{year}{2016}\natexlab{}.
\newblock \showarticletitle{Differentially private oblivious ram}.
\newblock \bibinfo{journal}{\emph{arXiv preprint arXiv:1601.03378}} (\bibinfo{year}{2016}).
\newblock


\bibitem[\protect\citeauthoryear{Wagh, He, Machanavajjhala, and Mittal}{Wagh et~al\mbox{.}}{2021}]%
        {wagh2021dp}
\bibfield{author}{\bibinfo{person}{Sameer Wagh}, \bibinfo{person}{Xi He}, \bibinfo{person}{Ashwin Machanavajjhala}, {and} \bibinfo{person}{Prateek Mittal}.} \bibinfo{year}{2021}\natexlab{}.
\newblock \showarticletitle{DP-cryptography: marrying differential privacy and cryptography in emerging applications}.
\newblock \bibinfo{journal}{\emph{Commun. ACM}} \bibinfo{volume}{64}, \bibinfo{number}{2} (\bibinfo{year}{2021}), \bibinfo{pages}{84--93}.
\newblock


\bibitem[\protect\citeauthoryear{Wang, Bater, Nayak, and Machanavajjhala}{Wang et~al\mbox{.}}{2021}]%
        {wang2021dp}
\bibfield{author}{\bibinfo{person}{Chenghong Wang}, \bibinfo{person}{Johes Bater}, \bibinfo{person}{Kartik Nayak}, {and} \bibinfo{person}{Ashwin Machanavajjhala}.} \bibinfo{year}{2021}\natexlab{}.
\newblock \showarticletitle{DP-Sync: Hiding update patterns in secure outsourced databases with differential privacy}. In \bibinfo{booktitle}{\emph{Proceedings of the 2021 International Conference on Management of Data}}. \bibinfo{pages}{1892--1905}.
\newblock


\bibitem[\protect\citeauthoryear{Wang, Bater, Nayak, and Machanavajjhala}{Wang et~al\mbox{.}}{2022}]%
        {wang2022incshrink}
\bibfield{author}{\bibinfo{person}{Chenghong Wang}, \bibinfo{person}{Johes Bater}, \bibinfo{person}{Kartik Nayak}, {and} \bibinfo{person}{Ashwin Machanavajjhala}.} \bibinfo{year}{2022}\natexlab{}.
\newblock \showarticletitle{IncShrink: Architecting Efficient Outsourced Databases using Incremental MPC and Differential Privacy}.
\newblock \bibinfo{journal}{\emph{arXiv preprint arXiv:2203.05084}} (\bibinfo{year}{2022}).
\newblock


\bibitem[\protect\citeauthoryear{Wang, Pujo, Nayak, and Machanavajjhala}{Wang et~al\mbox{.}}{2023}]%
        {wang2023private}
\bibfield{author}{\bibinfo{person}{Chenghong Wang}, \bibinfo{person}{David Pujo}, \bibinfo{person}{Kartik Nayak}, {and} \bibinfo{person}{Ashwin Machanavajjhala}.} \bibinfo{year}{2023}\natexlab{}.
\newblock \showarticletitle{Private Proof-of-Stake Blockchains using Differentially-private Stake Distortion}.
\newblock \bibinfo{journal}{\emph{Cryptology ePrint Archive}} (\bibinfo{year}{2023}).
\newblock


\bibitem[\protect\citeauthoryear{Wang, Nayak, Liu, Chan, Shi, Stefanov, and Huang}{Wang et~al\mbox{.}}{2014}]%
        {wang2014oblivious}
\bibfield{author}{\bibinfo{person}{Xiao~Shaun Wang}, \bibinfo{person}{Kartik Nayak}, \bibinfo{person}{Chang Liu}, \bibinfo{person}{TH~Hubert Chan}, \bibinfo{person}{Elaine Shi}, \bibinfo{person}{Emil Stefanov}, {and} \bibinfo{person}{Yan Huang}.} \bibinfo{year}{2014}\natexlab{}.
\newblock \showarticletitle{Oblivious data structures}. In \bibinfo{booktitle}{\emph{Proceedings of the 2014 ACM SIGSAC Conference on Computer and Communications Security}}. \bibinfo{pages}{215--226}.
\newblock


\bibitem[\protect\citeauthoryear{Wang and Yi}{Wang and Yi}{2021}]%
        {wang2021secure}
\bibfield{author}{\bibinfo{person}{Yilei Wang} {and} \bibinfo{person}{Ke Yi}.} \bibinfo{year}{2021}\natexlab{}.
\newblock \showarticletitle{Secure yannakakis: Join-aggregate queries over private data}. In \bibinfo{booktitle}{\emph{Proceedings of the 2021 International Conference on Management of Data}}. \bibinfo{pages}{1969--1981}.
\newblock


\bibitem[\protect\citeauthoryear{Xiao and Xiong}{Xiao and Xiong}{2015}]%
        {xiao2015protecting}
\bibfield{author}{\bibinfo{person}{Yonghui Xiao} {and} \bibinfo{person}{Li Xiong}.} \bibinfo{year}{2015}\natexlab{}.
\newblock \showarticletitle{Protecting locations with differential privacy under temporal correlations}. In \bibinfo{booktitle}{\emph{Proceedings of the 22nd ACM SIGSAC Conference on Computer and Communications Security}}. \bibinfo{pages}{1298--1309}.
\newblock


\bibitem[\protect\citeauthoryear{Yao}{Yao}{1986}]%
        {yao1986generate}
\bibfield{author}{\bibinfo{person}{Andrew Chi-Chih Yao}.} \bibinfo{year}{1986}\natexlab{}.
\newblock \showarticletitle{How to generate and exchange secrets}. In \bibinfo{booktitle}{\emph{27th annual symposium on foundations of computer science (Sfcs 1986)}}. IEEE, \bibinfo{pages}{162--167}.
\newblock


\bibitem[\protect\citeauthoryear{Zhang, Bater, Nayak, and Machanavajjhala}{Zhang et~al\mbox{.}}{2023}]%
        {zhang2023longshot}
\bibfield{author}{\bibinfo{person}{Yanping Zhang}, \bibinfo{person}{Johes Bater}, \bibinfo{person}{Kartik Nayak}, {and} \bibinfo{person}{Ashwin Machanavajjhala}.} \bibinfo{year}{2023}\natexlab{}.
\newblock \showarticletitle{Longshot: Indexing Growing Databases Using MPC and Differential Privacy}.
\newblock \bibinfo{journal}{\emph{Proceedings of the VLDB Endowment}} \bibinfo{volume}{16}, \bibinfo{number}{8} (\bibinfo{year}{2023}), \bibinfo{pages}{2005--2018}.
\newblock


\bibitem[\protect\citeauthoryear{Zhang, Katz, and Papamanthou}{Zhang et~al\mbox{.}}{2016}]%
        {zhang2016all}
\bibfield{author}{\bibinfo{person}{Yupeng Zhang}, \bibinfo{person}{Jonathan Katz}, {and} \bibinfo{person}{Charalampos Papamanthou}.} \bibinfo{year}{2016}\natexlab{}.
\newblock \showarticletitle{All your queries are belong to us: the power of $\{$File-Injection$\}$ attacks on searchable encryption}. In \bibinfo{booktitle}{\emph{25th USENIX Security Symposium (USENIX Security 16)}}. \bibinfo{pages}{707--720}.
\newblock


\bibitem[\protect\citeauthoryear{Zhang, Wang, Li, Honorio, Backes, He, Chen, and Zhang}{Zhang et~al\mbox{.}}{2021}]%
        {zhang2021privsyn}
\bibfield{author}{\bibinfo{person}{Zhikun Zhang}, \bibinfo{person}{Tianhao Wang}, \bibinfo{person}{Ninghui Li}, \bibinfo{person}{Jean Honorio}, \bibinfo{person}{Michael Backes}, \bibinfo{person}{Shibo He}, \bibinfo{person}{Jiming Chen}, {and} \bibinfo{person}{Yang Zhang}.} \bibinfo{year}{2021}\natexlab{}.
\newblock \showarticletitle{$\{$PrivSyn$\}$: Differentially Private Data Synthesis}. In \bibinfo{booktitle}{\emph{30th USENIX Security Symposium (USENIX Security 21)}}. \bibinfo{pages}{929--946}.
\newblock


\bibitem[\protect\citeauthoryear{Zheng, Dave, Beekman, Popa, Gonzalez, and Stoica}{Zheng et~al\mbox{.}}{2017}]%
        {zheng2017opaque}
\bibfield{author}{\bibinfo{person}{Wenting Zheng}, \bibinfo{person}{Ankur Dave}, \bibinfo{person}{Jethro~G Beekman}, \bibinfo{person}{Raluca~Ada Popa}, \bibinfo{person}{Joseph~E Gonzalez}, {and} \bibinfo{person}{Ion Stoica}.} \bibinfo{year}{2017}\natexlab{}.
\newblock \showarticletitle{Opaque: An oblivious and encrypted distributed analytics platform}. In \bibinfo{booktitle}{\emph{14th USENIX Symposium on Networked Systems Design and Implementation (NSDI 17)}}. \bibinfo{pages}{283--298}.
\newblock


\end{thebibliography}

\end{document}